\documentclass[pra,twocolumn]{revtex4-1}
\usepackage{amssymb,latexsym}
\usepackage{amsbsy} % for \boldsymbol
\usepackage{amsmath,amsthm}
\usepackage{hyperref}
\usepackage{graphicx,color}
\usepackage{verbatim}

\newtheorem{theorem}{Theorem}[section]
\newtheorem{lemma}{Lemma}[section]

\newtheorem{definition}{Definition}[section]

\newcommand\braket[1] {\langle {#1} \rangle}
\newcommand\bra[1] {\langle {#1} |}
\newcommand\ket[1] {| {#1} \rangle}

\newcommand\jpcomp[1] {j_{\mathrm{p};{#1}}}
\newcommand\jpvec {{\mathbf{j}_{\mathrm{p}}}}
\newcommand\jpvecix[1] {\mathbf{j}_{\mathrm{p};{#1}}}

\newcommand{\tr}{\operatorname{Tr}}
\newcommand{\diag}{\operatorname{diag}}
\renewcommand{\vec}[1]{\mathbf{#1}}
\newcommand{\RR}{\mathbb{R}}
\newcommand{\CC}{\mathbb{C}}
\newcommand{\rmi}{\mathrm{i}}
\newcommand{\loc}{{\mathrm{loc}}}

\renewcommand{\Re}{\operatorname{Re}}
\renewcommand{\Im}{\operatorname{Im}}

\begin{document}

\title{Fermion $N$-representability for prescribed density and paramagnetic current density}

\author{Erik Tellgren}
\email{erik.tellgren@kjemi.uio.no}
\affiliation{University of Oslo, Centre for Theoretical and
  Computational Chemistry,
  N-0315 Oslo, Norway}
\author{Simen Kvaal}
%\email{simen.kvaal@kjemi.uio.no}
\affiliation{University of Oslo, Centre for Theoretical and
  Computational Chemistry,
  N-0315 Oslo, Norway}
\author{Trygve Helgaker}
%\email{t.u.helgaker@kjemi.uio.no}
\affiliation{University of Oslo, Centre for Theoretical and
  Computational Chemistry,
  N-0315 Oslo, Norway}

\begin{abstract}
  The $N$-representability problem is the problem of determining
  whether or not there exists $N$-particle states with some prescribed
  property. Here we report an affirmative solution to the fermion
  $N$-representability problem when both the density and paramagnetic
  current density are prescribed. This problem arises in
  current-density functional theory and is a generalization of the
  well-studied corresponding problem (only the density prescribed) in
  density functional theory. Given any density and paramagnetic
  current density satisfying a minimal regularity condition
  (essentially that a von Weiz\"acker-like the canonical kinetic
  energy density is locally integrable), we prove that there exist a
  corresponding $N$-particle state. We prove this by constructing an
  explicit one-particle reduced density matrix in the form of a
  position-space kernel, i.e.\ a function of two continuous position
  variables.
%  When such a kernel is regarded as an element of the
%  Lebesgue space $L^2$, its diagonal corresponds to a set of measure
%  zero and the function values on the diagonal are therefore not in
%  general meaningful, making it non-trivial to verify to what extent
%  conventional expressions for the density and paramagnetic current
%  density are meaningful.
  In order to make minimal assumptions, we
  also address mathematical subtleties regarding the diagonal of, and
  how to rigorously extract paramagnetic current densities from,
  one-particle reduced density matrices in kernel form.
\end{abstract}

\maketitle 

\section{Introduction}

The question of $N$-representability has been studied extensively in
quantum chemistry and related fields~\cite{COLEMAN_RMP35_668}. In
particular, it plays an important role in density-functional
theory (DFT). Given prescribed values for quantities in a fermionic system, e.g.\ its electron density or
its reduced density matrix, one may ask whether or not it can be obtained
from a Slater determinant, from a pure $N$-particle state, or from a mixed
$N$-particle state. Regarding the density, it is well known that any density with
a finite von Weizs\"acker kinetic energy may be reproduced using a
Slater determinant that also 
has finite kinetic energy~\cite{MACKE_PR100_992,GILBERT_PRB12_2111,HARRIMAN_PRA24_680,ZUMBACH_PRA28_544,GHOSH_JCP82_3307}.
For a one-particle reduced density
matrix (1-rdm), Slater-determinantal representability is equivalent
to idempotency; in general, however, pure-state $N$-representability of 1-rdms is a difficult and
largely unsolved
problem~\cite{RUSKAI_PR169_101,LUDENA_JMST123_371,RUSKAI_JPA40_F961}. On the other hand,
any 1-rdm with spin-orbital occupation numbers (eigenvalues) in
the range $[0,1]$ and trace $N$ may be obtained from a mixed
$N$-particle state.

In this note, we report the solution to the mixed-state
$N$-representability problem when both the density and paramagnetic
current density are prescribed. More precisely, we answer the following question: given a density
$\rho(\mathbf{r})$ and a paramagnetic current density $\jpvec(\mathbf{r})$, does
there exist a mixed state $\Gamma$ with the prescribed density and
current density, written $\Gamma\mapsto(\rho,\jpvec)$? We answer this question affirmatively by 
constructing an explicit 1-rdm, from which the existence of the $N$-particle state
follows.

We note that standard constructions demonstrating the
Slater-determinantal $N$-representability when only the density is
prescribed rely on the use of equidensity orbitals. Also, early work by Ghosh and
Dhara~\cite{GHOSH_PRA38_1149} sketched a construction of such
equidensity orbitals that reproduce  both densities and currents. However,
such solutions have limited scope, since the vorticity vanishes when
orbitals give rise to the same density.

During the preparation of the present work, one of us (S.~Kvaal) met
E.H.~Lieb during a trimester at Institut Henri Poincar\'e in Paris in
July 2013, and became aware that together with R.~Schrader, he had
shown a Slater determinant representability result for $(\rho,\jpvec)$
and $N\geq 4$. This work has now been published
\cite{LIEB_SCHRADER_2013}.

Clearly, $N$-representability via a Slater determinant implies
representabiility via a mixed state. However, Lieb and Schrader's
result requires $N\geq 4$, and they also give a counterexample for
$N=2$, where no Slater determinant can exist (with continuously
differentiable and single-valued orbital phase functions). Our result,
while showing a weaker sense of $N$-representability, has no condition
on $N$. Moreover, both the present work and
Ref.~\cite{LIEB_SCHRADER_2013} have mild regularity and decay
assumptions on $(\rho,\jpvec)$ that ensure representability, but these
are different in the two approaches. The techniques of proof are also otherwise
significantly different: Lieb and Schrader rely on the so-called
smooth Hobby--Rice theorem, while our approach is by
direct construction of a 1-rdm. The present work and
the work of Lieb and Schrader are complementary, offering two
different points of view and solutions to a long-standing problem.

%
%In
%our work, there are no conditions on $N$. Moreover, Lieb and Schrader
%point out the existence of $N=2$ cases where \emph{no} Slater
%determinants exist. Their non-existence proof uses the additional
%assumption that orbital phase functions are continuously
%differentiable, leaving it open whether Slater determinants built from
%orbitals with multi-valued phase functions (e.g.\ phase functions
%arising from spherical harmonics) exist. The present work and
%Ref.~\cite{LIEB_SCHRADER_2013} are therefore complementary.

The remainder of this paper contains five sections. In Section\;\ref{background}, we give some background information and establish notation.
Following a discussion of the relationship between a reduced density matrix
and its associated density and paramagnetic current
density in Section\;\ref{secDIAG}, we construct in Sections\;\ref{redmat} and~\ref{cankin} a reduced density matrix 
for a prescribed density and paramagnetic current density. 
Section\;\ref{discussion} contains some concluding remarks.
Finally, two appendices are also provided. Appendix~A contains a brief overview
of some mathematical concepts and results on Hilbert--Schmidt
operators needed for the main results
of Section~\ref{secDIAG}. Appendix~B contains proofs of theorems in
Section~\ref{secDIAG}.

\section{Background}
\label{background}

%The $N$-particle density--current-density  representability  problem
The $N$-representability problem with prescribed density $\rho$ and
paramagnetic current density $\jpvec$ arises in current-density
functional theory (CDFT)~\cite{VIGNALE_PRL59_2360}. In CDFT, a
magnetic vector potential $\mathbf{A}$, in addition the scalar potential $v$,
enters the (spin-free) $N$-electron Hamiltonian. In atomic units,
\begin{align}
 H[v,\mathbf{A}] &= \frac{1}{2}\sum_{k=1}^N (- \mathrm i\boldsymbol \nabla_{k} + \mathbf{A}(\mathbf{r}_k))^2 \nonumber \\ & \quad\quad\quad + \sum_{k=1}^N v(\mathbf{r}_k) + \sum_{k<l} \frac{1}{r_{kl}}.
\end{align}
Here $\mathbf r_k$ is the position of electron $k$, the operator $\boldsymbol \nabla_k$ differentiates with respect to $\mathbf r_k$, and
$r_{kl}$ is the distance between electrons $k$ and $l$.
The corresponding ground-state energy is given by the Rayleigh--Ritz variation principle,
\begin{equation}
  E[v,\mathbf{A}] = \inf_{\Gamma} \mathrm{Tr}( \Gamma H[v,\mathbf{A}])
%   = \inf_{\rho,\jpvec} \left( 
%\int \! \Bigl (\rho (v + \tfrac{1}{2} A^2) +
%     \jpvec\cdot\mathbf{A}\Bigr) \, \mathrm d\mathbf{r} 
%+ \inf_{\Gamma \mapsto \rho,\jpvec} \mathrm{Tr}( \Gamma H[0,\mathbf{0}]) 
%\right),
   \label{eq0:energy}
\end{equation}
where the minimization is over all mixed states $\Gamma$ with
a finite canonical kinetic energy 
\begin{equation}
T[\Gamma] := \frac{1}{2} \tr \left(\boldsymbol \nabla\Gamma\boldsymbol \nabla^\dag \right).
\end{equation}
Introducing the constrained-search universal functional
\begin{equation}
F[\rho,\jpvec] =
\inf_{\Gamma \mapsto \rho,\jpvec} \mathrm{Tr} \left( \Gamma H[0,\mathbf{0}] \right)  ,
\end{equation}
we may rewrite the Rayleigh--Ritz variation principle in Eq.\,(\ref{eq0:energy}) in the form of a Hohenberg--Kohn variation principle,
\begin{equation}
  E[v,\mathbf{A}] 
   = \inf_{\rho,\jpvec} \Bigl( \! F[\rho,\jpvec] + 
\int \! \Bigl (\rho \,(v + \tfrac{1}{2} A^2) +
     \jpvec\cdot\mathbf{A}\Bigr)  \mathrm d\mathbf{r}  \Bigr).
   \label{eq:energy}
\end{equation}
The mixed-state $N$-representability problem is
directly related to how large the search domain in~Eq.\;(\ref{eq:energy}) needs to be: if no 
$\Gamma\mapsto(\rho,\mathbf{j}_\text{p})$ exists, then
$F[\rho,\jpvec]=+\infty$ by definition. 

In Kohn--Sham
theory, the idea is to express the densities in Eq.\,(\ref{eq:energy}) in terms of
a single Slater determinant of non-interacting
particles and to approximate the kinetic-energy contributions to
$F[\rho,\jpvec]$ by the non-interacting kinetic energy,
\begin{equation}
  T_s[\rho,\jpvec] = \inf_{\{\phi_k\}_{k=1}^N \mapsto \rho,\jpvec}
  \frac{1}{2}\sum_{k=1}^N \braket{\boldsymbol \nabla\phi_k, \boldsymbol \nabla\phi_k},
\label{sdnrep}
\end{equation}
where the infimum is over an orthonormal set of orbitals $\phi_k$ or, equivalently, the
corresponding Slater determinants or idempotent 1-rdms. 
At this point, the
Slater-determinantal $N$-representability problem arises. 

In general,
densities $\rho$ and $\jpvec$ arising from a single orbital have
a vanishing paramagnetic vorticity,
\begin{equation}
  \boldsymbol{\nu} = \boldsymbol \nabla\times\frac{\jpvec}{\rho} = 0,
\end{equation}
except for possible Dirac-delta singularities at points $\mathbf r$ where $\rho(\mathbf r) =
0$. Consequently, a closed-shell two-particle Kohn--Sham system 
can only reproduce paramagnetic densities
with vanishing vorticity. In general, therefore an extended Kohn--Sham approach
with fractional occupation numbers is required (see
Refs.\;\cite{CANCES_JCP114_10616,CANCES_JCP118_5364,KRAISLER_PRA80_032115,NYGAARD_JCP138_094109}
for work in this direction),
\begin{equation}
  \bar{T}_s[\rho,\jpvec] = \inf_{\{n_k\phi_k\} \mapsto \rho,\jpvec}
  \frac{1}{2} \sum_{k=1}^{\infty} n_k  \braket{\boldsymbol \nabla\phi_k,\boldsymbol \nabla\phi_k},
\end{equation}
where orthonormality, $0 \leq n_k \leq 1$, and $\sum_k n_k = N$ are
additional constraints on the infimum. Alternatively, since $n_k$ and
$\phi_k$ are eigenvalues and eigenvectors of 1-rdms, the minimization may equivalently 
be performed over 1-rdms. Here, the mixed-state $N$-representability problem appears.

For a mixed state 
$\Gamma\mapsto(\rho,\jpvec)$, it is known that
\begin{align}
  T_{\text{W}}[\rho] &+ T_\text{p}[\rho,\mathbf{j}_\text{p}] \leq T[\Gamma],
\end{align}
where the von Weizs{\"a}cker kinetic-energy functionals are given by
\begin{align}
  T_{\text{W}}[\rho]  &:= \frac{1}{8}\int \!\! \rho(\mathbf{r})^{-1}|\boldsymbol \nabla\rho(\mathbf{r})|^2 \mathrm d\mathrm{r} , \\
  T_\text{p}[\rho,\mathbf{j}_\text{p}] &:= 
  \frac{1}{2} \int \!\! \rho(\mathbf{r})^{-1}|\mathbf{j}_\text{p}(\mathbf{r})|^2 \mathrm d\mathrm{r} . 
\end{align}
A necessary condition for a finite-kinetic-energy representability is therefore that $T_W[\rho] +
T_\text{p}[\rho,\jpvec]<+\infty$. For the case $\jpvec=0$, this is
also a sufficient condition. It is of interest to know whether this sufficiency 
generalizes to $\jpvec\neq 0$. In this paper, we shall prove sufficiency under mild additional
conditions on the current density.

\section{Diagonals of density operators}
\label{secDIAG}

%From now on, all $N$-electron wave functions $\Psi\in L^2(\RR^{3N})$ are assumed
%to be properly antisymmetrized and normalized, and all mixed states
%$\Gamma = \sum_i p_i \ket{\Psi_i}\bra{\Psi_i}$ are assumed to be built
%from properly antisymmetrized and orthonormal pure $N$-electron
%states. Furthermore, whenever we make reference to a compact subset
%$K\subset \RR^3$, it is assumed that it is also measurable.

We do not explicitly consider spin and therefore take as our point
of departure an $N$-electron density matrix that depends only on spatial coordinates, with  the spin coordinates integrated out:
\begin{equation}
  \label{eqGAMMAINTRO}
  \Gamma(\mathbf{r}_{1:N},\mathbf{s}_{1:N}) = \sum_i p_i \Psi_i(\mathbf{r}_{1:N}) \Psi^\ast_i(\mathbf{s}_{1:N}).
\end{equation}
Such a density matrix is an element of a
Lebesgue space, 
\begin{equation}
\Gamma \in L^2(\RR^{3N}\times\RR^{3N}), 
\end{equation}
and is symmetric with respect to permutations $\pi$ of the
coordinate labels, $(\mathbf{r}_k,\mathbf{s}_k)
\stackrel{\pi}{\mapsto}
(\mathbf{r}_{\pi(k)},\mathbf{s}_{\pi(k)})$. 
The right-hand side of Eq.~\eqref{eqGAMMAINTRO} is a convex combination of properly normalized pure states, $\Psi_i
\in L^2(\RR^{3N})$, with coefficients $p_i \geq 0$ such that $\sum_i p_i = 1$.
Moreover, each (spin-free) pure state is either totally symmetric or anti-symmetric. 
In what follows, it does not matter whether each pure state is required to be
anti-symmetric, symmetric, or either anti-symmetric or
symmetric with respect to index permutations $\pi$ of the spatial coordinates. For simplicity, we
shall occasionally simplify the presentation by taking $\Gamma$ to be a
pure state.

\subsection{Density matrices}
\label{denmats}

The 1-rdm belonging to a pointwise defined $\Gamma$ on the form~\eqref{eqGAMMAINTRO} is given by the convex combination
\begin{equation}
  D_{\Gamma}(\vec{r},\vec{s}) := N \sum_i p_i  \!\!\int_{\RR^{3N-3}} 
  \!\!\!\!
  \!\!\!\!
  \!\!\!\!
\Psi_i(\vec{r},\vec{r}_{2:N}) \Psi^\ast_i(\vec{s},\vec{r}_{2:N}), 
\mathrm d\vec{r}_{2:N}
  \label{eq:Dpsi}
\end{equation}
and belongs to  $L^2(\RR^3\times\RR^3)$. 
%(See Appendix A for a review of Lebesgue spaces, etc.) 
For pure states $\Gamma = \ket{\Psi}
\bra{\Psi}$, we may alternatively write $D_{\Psi}$. Due to permutation
symmetry, the 1-rdm is independent of which $N-1$ coordinates that have been 
integrated out.  

Given that $D_\Gamma \in L^2(\RR^3\times\RR^3)$, $D_\Gamma$ is by 
definition the \emph{kernel of a Hilbert--Schmidt integral operator}.
%moreover, since $\tr  D_\Gamma =
%\int \! D_\Gamma(\mathbf r,\mathbf r )\,\mathrm d \mathbf r = 
%\int \! \rho_\Gamma(\mathbf r)\,\mathrm d \mathbf r = N$ is finite,
%$D_\Gamma$ is more strongly the \emph{kernel of a trace-class operator}. 
Our discussion makes extensive use of this basic fact about the reduced
density matrix. In particular, 
$\Gamma$ and $D_\Gamma$ are both \emph{trace-class operators}---that is, Hilbert--Schmidt operators for
which the matrix trace has a meaningful generalization. Let
$\{\phi_k\}\subset L^2(X)$ be an orthonormal basis. By definition, $A$
is a trace-class operator if and only if the trace
\begin{equation}
\tr A := \sum_k \braket{\phi_k, A\phi_k} 
\end{equation}
has a finite value, independent of the orthonormal basis.
 
For two Hilbert--Schmidt operators $B$ and $C$, the kernel of the operator product $A = B \ast C$ is easily seen to be
\begin{equation}
  (B*C)(x,y) := \int_X \!\!
  B(x,z)C(z,y) \,\mathrm d z, \label{eq:star}
\end{equation}
which is also Hilbert--Schmidt. Importantly, 
it can be shown that $A$ is (the kernel of) a trace-class operator if
and only if $A = B\ast C$ with $B$ and $C$ Hilbert--Schmidt. (Indeed,
this is often taken as an alternative definition of trace-class operators.) 
The trace is then given by \cite{BRISLAWN_1988} the integral of the diagonal,
\begin{equation}
  \tr A = \int_{X} (A\ast B)(x,x)\, \mathrm{d} x.
\end{equation}
If $A$ is diagonable (e.g., symmetric positive
semidefinite), then $\tr A$ is the sum of the eigenvalues, like in the
finite-dimensional case. For further information on these operator
classes, see for example the standard textbook~\cite{REED_SIMON_1980}.

%\emph{Hilbert--Schmidt} operators over $L^2(X)$ are precisely those
%operators that can be defined in terms of integral kernels in
%$L^2(X\times X)$. 

We denote by $\mathcal{D}_N$ the set of mixed $N$-electron states
$\Gamma$ and by $\mathcal{D}_{N,1}$ the set of 1-rdms that belong to
some mixed $N$-electron state. 
The set $\mathcal{D}_{N,1}$ has the following
well-known characterization:
\begin{theorem}\label{thm:DN1properties}%[Properties of $\mathcal{D}_{N,1}$]
  $\mathcal{D}_{N,1}$ consists of those $D\in L^2(\RR^3\times
  \RR^3)$ with the following properties:
  \begin{enumerate}
    \item\label{item:prop1} $D$ is the kernel of a trace-class operator on 
      $L^2(\RR^3)$.
    \item\label{item:prop2}
      $D$ is Hermitian: \newline $D(\vec{r},\vec{s}) = D^\ast(\vec{s},\vec{r})$
      for almost all $(\vec{r},\vec{s})$.
    \item\label{item:prop3}
      $D$ is positive semidefinite: \newline $0 \leq %n(\mathbf r) =  
\int \! \phi^\ast(\vec{r})
      D(\vec{r},\vec{s})\phi(\vec{s}) \mathrm d\vec{s}$ for all $\phi\in L^2(\RR^3)$
    \item\label{item:prop4}
      $D$ has no eigenvalues greater than two: \newline $2 \geq \int \! \phi^\ast(\vec{r})
      D(\vec{r},\vec{s})\phi(\vec{s}) \mathrm d\vec{s}$ for all $\phi\in
      L^2(\RR^3)$
    \item\label{item:prop5}
      $D$ has eigenvalues that add up to  $N$: \newline$\tr D = N$.
  \end{enumerate}
\end{theorem} 
\begin{proof}
  See Ref.~\cite{PARR_YANG_1989}, Section 2.6.
\end{proof}

The last three conditions mean that $D(\mathbf{r},\mathbf{s})$ has
eigenvalues in the interval $[0,2]$---that is, eigenvalues interpretable
as fermion occupation numbers---and that the sum of the eigenvalues
$\tr D$ is equal to $N$, the number of particles.

%In the following, we shall treat $D\in \mathcal{D}_{N,1}\subset
%L^2(\RR^3\times\RR^3)$ as the kernel of a Hilbert--Schmidt operator
%over $L^2(\RR^3)$. (We recall that an integral operator
%is a Hilbert--Schmidt integral operator if it has a kernel in $L^2(\RR^3\times\RR^3)$.)

% We remark that $A\in L^2(\RR^3\times \RR^3)$ is
% the kernel of a trace-class operator if and only if it can be written
% in the product form $A = B\ast C$ where
% \begin{equation}
%   (B*C)(\vec{r},\vec{s}) := \int_{\RR^3} \!\!
%   B(\vec{r},\vec{u})C(\vec{u},\vec{s}) \,\mathrm d \vec u \label{eq:star}
% \end{equation}
% with $B,C$ in $L^2(\RR^3\times\RR^3)$, meaning that $B,C$ are kernels of
% Hilbert--Schmidt operators over $L^2(\RR^3)$. Indeed, this requirement is often taken as
% the definition of trace-class operators. It is a general result that,
% if an operator $A=B\ast C$ is trace class, then~\cite{BRISLAWN_1988}
% \begin{equation}
%   \tr A = \int_{\RR^3} (B*C)(\vec{r},\vec{r}) \, \mathrm d \mathbf r .
% \end{equation}
Since $D\in\mathcal{D}_{N,1}$ is Hermitian and positive, it is
easy to show that there always exists a factorization of the form $D = G^\dag
\ast G$, meaning that we may write the density matrix in the form
\begin{equation}
    D(\vec{r},\vec{s}) = (G^\dag * G)(\vec{r},\vec{s}) =
    \int_{\RR^3} \! G^\ast(\vec{u},\vec{s})
    G(\vec{u},\vec{r}) \mathrm d\vec{u},
\end{equation}
which plays an important role in the following.

%We remark that the number of particles need not be an
%integer, i.e., $D$ may be the 1-rdm of an arbitrary Fock space
%density operator.

\subsection{Density}

We now define the \emph{density} $\rho_\Psi$ associated with the wave function $\Psi$ as
\begin{equation}
  \rho_\Psi(\vec{r}) := D_\Psi(\vec{r},\vec{r}) =  N\!\!\int_{\RR^{3N-3}}  \!\!
  |\Psi(\vec{r},\vec{r}_{2:N})|^2 
\mathrm d\vec{r}_{2:N}.
  \label{eq:rhopsi}
\end{equation}
For almost all $\vec{r}$, it holds that $\Psi(\vec{r},\cdot)\in
L^2(\RR^{3N-3})$. Using the Cauchy--Schwarz inequality, we see from
Eq.\;\eqref{eq:rhopsi} that $\rho_\Psi(\vec{r}) =
D_\Psi(\vec{r},\vec{r})$ is well defined for almost all $\vec{r}$.
For a mixed state $\Gamma \in \mathcal D_N$, the density $\rho_\Gamma$ is defined in the same manner but from $D_\Gamma$.

%There is a fine point to be made a big deal of in this section.
The following point is subtle but important here.
We write $\Gamma\mapsto D$ whenever 
$\|D_\Gamma-D\|_{L^2(\RR^3\times\RR^3)} = 0$.
%meaning that $D_\Gamma=D$ almost everywhere.) 
This statement does not imply that 
that $D_\Gamma = D$ everywhere, only that $D_\Gamma$ and $D$ are equal as 
elements of $L^2(\RR^3\times \RR^3)$.
Consequently, $D_\Gamma$ and $D$
may differ at a set of measure zero, \emph{including the totality of the diagonal}.
%Therefore, if $D$ is given as a \emph{pointwise function}, 
%then it may in principle differ from $D_\Gamma$ everywhere on the diagonal.
%is not straightforward that the diagonal values are actually equal to
%those of $D_\Gamma$: $D_\Gamma \neq D$ \emph{as pointwise functions}, but
%they are equal as elements of $L^2(\RR^3\times \RR^3)$, i.e., 
%$\|D_\Gamma-D\|_{L^2(\RR^3\times\RR^3)} = 0$. They  may differ at a set of measure zero,
%including the totality of the diagonal.
Therefore, we need to examine carefully the validity or meaning of the statement
``$\rho_\Gamma(\vec{r}) = D(\vec{r},\vec{r})$'' for a state and density matrix related by $\Gamma\mapsto D$.

Suppose next that we are able to assign a diagonal $\diag D$ to $D$ in some
unambiguous way and let $\Gamma,\Gamma'\in\mathcal{D}_N$ be two (possibly
distinct) states such that $\Gamma \mapsto D$ and $\Gamma' \mapsto D$, meaning
that $D = D_\Gamma =
D_{\Gamma'}$ almost everywhere in $\mathbb R^3 \times \mathbb R^3$. Is it then true
that $\rho_\Gamma = \rho_{\Gamma'} = \diag D$ almost everywhere in $\mathbb R^3$?
Intuitively, this should be so.

The following theorem, which is proved in Appendix\;B, resolves the issue: 
%The first result deals with extracting a unique density from $D\in
%\mathcal{D}_{N,1}$, showing that this agrees with $\rho_\Gamma$ for
%any $\Gamma\in\mathcal{D}_N$ such that $\Gamma\mapsto D$.
%
\begin{theorem}\label{thm:densities}
  Let $D\in \mathcal{D}_{N,1}$, and suppose that  $G\in
  L^2(\RR^3\times\RR^3)$ is such that
  \begin{equation*}
    D(\vec{r},\vec{s}) = (G^\dag * G)(\vec{r},\vec{s})
  \end{equation*}
  almost everywhere in $\RR^3\times\RR^3$.
%  i.e.,
%  \begin{equation}
%    D(\vec{r},\vec{s}) = \int_{\RR^3} G(\vec{u},\vec{r})^*
%    G(\vec{u},\vec{s}) d\vec{u} \quad \mathrm{a.e.} [\mathrm d\vec{r}d\vec{s}]
%  \end{equation}
  Then, for every $\Gamma\in \mathcal{D}_N$ such that $\Gamma\mapsto
  D$, it holds that
  \begin{align*}
%    \rho_\Gamma(\vec{r}) &=  D(\mathbf r,\mathbf r) 
    \rho_\Gamma(\vec{r}) &=  (G^\dag*G)(\vec{r},\vec{r}) 
%    \int_{\RR^3} \! G^\ast(\vec{u},\vec{r})
%    G(\vec{u},\vec{r}) \mathrm d\vec{u}
%    \mathrm{a.e.}\ [\mathrm d\vec{r}],
  \end{align*}
  almost everywhere in $\RR^3$.
%  i.e.,
%  \begin{equation}
%    \rho_\Gamma(\vec{r}) = \int_{\RR^3} G(\vec{u},\vec{r})^*
%    G(\vec{u},\vec{r}) d\vec{u} \quad \mathrm{a.e.}\ [\mathrm d\vec{r}].
%  \end{equation}
\end{theorem}

%\begin{proof}
%  See Appendix A.
%\end{proof}

Since the factorization of $D = G^\dag\ast G$ does exist following the discussion in Section\;\ref{denmats}, 
it is indeed meaningful to talk about ``the density $\rho$ of $D$'' without reference to a specific
$\Gamma\mapsto D$:
\begin{equation}
  \rho_D(\vec{r}) = \diag D(\vec{r}) := (G^\dag\ast G)(\vec{r},\vec{r})
  \quad \text{a.e.}
\end{equation}
In particular, it follows that
\begin{equation}
\tr D = \int \!\rho_D(\vec{r})\,\mathrm d\vec{r}.
\end{equation}
We emphasize that
we only define the diagonal when a factorization is present, and that this
diagonal is independent of the factorization.
%The diagonal
%$\diag A$ is well-defined in a more general setting, see
%Theorem~\ref{thm:factorization} in Appendix A.

\subsection{Momentum density}

Before considering the momentum density, we note that all derivatives that occur in the subsequent treatment 
are \emph{distributional or weak derivatives}. A function 
$f\in L^p(X)$, with $X\subset\RR^n$ open, is said to have a weak derivative $g = \partial_\alpha f
\in L^1_\loc(X)$ if, for all smooth, compactly supported ``test functions'' $u \in
\mathcal{C}^\infty_\mathrm c(X)$,
\begin{equation}
\int_X g(x)u(x) \mathrm dx = -\int_X f(x) \partial_\alpha u(x) \mathrm dx.
\end{equation} 
Thus, the weak derivative acts just like the standard derivative $\partial f/\partial x_\alpha$ 
when we apply integration by parts, coinciding with the
classical derivative whenever this exists. Higher-order weak derivatives
are defined in a similar manner. A standard monograph for weak
derivatives is Ref.~\cite{EVANS_2000}.

By analogy with the density in Eq.\;\eqref{eq:rhopsi}, we now define 
the \emph{momentum density} $\vec{c}_\Psi$ of a state $\Psi$ as
\begin{align}
  \vec{c}_\Psi(\vec{r}) &:= N \!\! \int_{\RR^{3N-3}} \!\!\!\!\!
  [-\rmi\nabla_{\vec{r}}\Psi(\vec{r},\vec{r}_{2:N})]\Psi^\ast(\vec{r},\vec{r}_{2:N})
\mathrm d\vec{r}_{2:N} 
\nonumber \\ 
&= -\rmi\nabla_\vec{r}D_\Psi(\vec{r},\vec{s})|_{\vec{r}=\vec{s}},
  \label{eq:cpsi}
\end{align}
whose real part is the paramagnetic current density:
\begin{equation}
\jpvec_\Psi(\vec{r}) = \Re \vec{c}_\Psi(\vec{r})
\end{equation}
with an analogous definitions for a mixed state $\Gamma$.
We note, however, that this definition may not make sense without additional
assumptions on the wave function $\Psi$, beyond those needed for the definition of the density.
We also observe that the second equality in Eq.\;\eqref{eq:cpsi}
needs to be justified further since 
$\nabla_\vec{r}D_\Psi(\vec{r},\vec{s})|_{\vec{r}=\vec{s}}$
is only defined pointwise almost everywhere 
and since integration may not commute with differentiation.

%We recall that the motivation for introducing $\rho_\Psi$ is that, if $v$ is a one-body potential with a finite expectation value
%$\braket{\Psi,\sum_i v(\vec{r}_i)\Psi}$, then this expectation value is equal to $\int \!\rho_\Psi(\vec{r})v(\vec{r}) \,\mathrm d \mathbf r$. 
%Likewise, if $\vec{A}$ is a vector potential such that
%$\sum_i\braket{\vec{A}(\vec{r}_i)\Psi,-\rmi\nabla_k\Psi}$ is finite,
%then this expectation value is equal to $\int \! \vec{c}_\Psi(\vec{r})\cdot\vec{A}(\vec{r})\,\mathrm d \mathbf r$,
%see Eq.\;\eqref{eq:energy}.

To assign unambiguously a momentum density
$\vec{c}_D(\vec{r})$ to $D\in\mathcal{D}_{N,1}$, we first introduce the
notion of a \emph{locally finite kinetic energy}:
%which essentially means
%that the kinetic energy density is in $L^1_\loc(\RR^3)$. 
%In the following definition, $\boldsymbol \nabla_i$ is the (weak) gradient acting on the $i$th coordinate.

\begin{definition}[Locally finite kinetic energy]
We say that $D\in\mathcal{D}_{N,1}$ has a locally finite kinetic
energy if the weak derivative $\boldsymbol \nabla_1\cdot\boldsymbol \nabla_2 D$
is the kernel of a trace class operator over $L^2(K)$ for every compact $K\subset
\RR^3$. Likewise, we say that $\Psi\in
L^2(\RR^{3N})$ has a locally finite kinetic energy if $\boldsymbol \nabla_1\Psi
\in L^2(K\times\RR^{3N-3})$ for every compact $K\subset\RR^3$---that is,
$\boldsymbol \nabla_1\Psi\in L^2(\RR^3_\loc\times\RR^{3N-3})$. (See
Appendix~\ref{secLOCINT}.)
A mixed state $\Gamma\in \mathcal{D}_{N}$ has a locally finite kinetic
energy if $\sum_i p_i
\|\boldsymbol \nabla_1\Psi_i\|^2_{L^2(K\times\RR^{3N-3})}$ is finite for every
compact $K\subset \RR^3$.
\end{definition}

Note that the pure-state definition of locally finite kinetic energy
follows from that of the mixed state. 
The various definitions of a locally finite kinetic energy are connected, as summarized in the
following theorem, proved in Appendix B:
\begin{theorem}\label{thm:lockin}
For $D\in\mathcal{D}_{N,1}$, the following statements are equivalent:
\begin{enumerate}
\item
$D$ has a locally finite kinetic energy.
\item
There exists a factorization $D = G^\dag \! \ast G$ (a.e.) with
$G$ Hilbert--Schmidt and 
%$\nabla_2 G\in L^2(\RR^3\times K)$ for every compact $K\subset\RR^3$, i.e., 
$\nabla_2 G\in L^2(\RR^3\times\RR^3_\loc)$.
\item
Any $\Gamma\in\mathcal{D}_{N}$ with $\Gamma\mapsto D$ has
a locally finite kinetic energy.
\end{enumerate}
\end{theorem}
%If a factorization $D = G^\dag\! *G$ is given with $\nabla_2 G \in
%L^2(\RR^3\times K)$ for every compact $K\subset\RR^3$, then $D$ has 
%a locally finite kinetic energy. 
If $D$ has a locally finite kinetic energy, then
the associated \emph{kinetic-energy density} is defined as 
\begin{align}
\tau_D(\vec{r}) &:= \frac{1}{2}\|\nabla_2
G(\cdot,\vec{r})\|^2_{L^2(\RR^3)} \notag \\ &=
\frac{1}{2}\int_{\RR^3} [\nabla_2
G(\vec{u},\vec{r})]^*\cdot [\nabla_2G(\vec{u},\vec{r})] \,\mathrm d \vec u\\
&=\frac{1}{2}\diag (\nabla_1\cdot\nabla_2 D)(\vec{r}),\notag
\end{align}
which is finite almost everywhere.
From the proof of the Theorem\;\ref{thm:lockin} in Appendix B, it follows that $\tau_D$ is in fact
the kinetic energy density of any $\Gamma\mapsto D$.  We also see that
$\tau_D\in L^1_\loc(\RR^3)$ and that the total kinetic energy is
finite if and only if $\tau_D\in L^1(\RR^3)$. 

Finally, the following theorem (proved in Appendix B) states that,
if $D$ has a locally finite kinetic energy, then the momentum density of~Eq.\;\eqref{eq:cpsi} is also well defined:
\begin{theorem}\label{thm:currents}
Let $D\in \mathcal{D}_{N,1}$ have a locally finite kinetic energy and
let $G\in L^2(\RR^3\times\RR^3)$ be such that $D = G^\dag*
G$ and $\nabla_2 G \in L^2(\RR^3\times \RR^3_\mathrm{loc})$. 
For each $\Gamma\in \mathcal{D}_N$ with $\Gamma\mapsto D$, it then holds that $\vec{c}_\Gamma \in
L^1_\loc(\RR^3)$ and that
\begin{align*}
\vec{c}_\Gamma(\vec{r}) &= \left([-\rmi\nabla_2 G]^\dag *
G\right)(\vec{r},\vec{r}) \quad\mathrm{a.e.} \notag\\
&= \diag (-\rmi\nabla_1 D)(\vec{r}) .
\end{align*}
\end{theorem}

This result implies that if $D$ has locally finite kinetic energy,
then $\jpvec_\Gamma\in L^1_\loc(\RR^3)$. Moreover, $\nabla\rho_\Gamma
= -2\Im \vec{c}_\Gamma \in L^1_\loc(\RR^3)$ as well.

\subsection{Summary}

For easy reference, we collect the main  conclusions of this section in a  
separate theorem:

\begin{theorem}\label{thm:summary}
  Let  $D = G^\dag\! \ast G$ with $G\in L^2(\RR^3\times\RR^3) \in \mathcal{D}_{N,1} $, $\nabla_2 G\in
  L^2(\RR^3\times \RR^3_\loc)$. For every
  $\Gamma\in \mathcal{D}_N$ such that $\Gamma\mapsto D$, it then holds that the density
  $\rho_\Gamma = \rho\in L^1(\RR^3)$, the momentum density $\vec{c}_\Gamma = \vec{c} \in L^1_\loc(\RR^3)$,
  and the kinetic energy density $\tau_\Gamma = \tau\in L^1_\loc(\RR^3)$ are given almost everywhere by
  the expressions
  \begin{align*}
    \rho(\vec{r}) &= \int_{\RR^3} G^\ast (\vec{u},\vec{r})
    G(\vec{u},\vec{s}) \,\mathrm d \vec u, \\
    \vec{c}(\vec{r}) &= \int_{\RR^3} [-\rmi\nabla_2
    G^\ast(\vec{u},\vec{r})] \,G(\vec{u},\vec{s})  \,\mathrm d \vec u,\\
    \tau(\vec{r}) &= \frac{1}{2} \int_{\RR^3} [\nabla_2
    G^\ast(\vec{u},\vec{r})] \cdot [\nabla_2 G(\vec{u},\vec{s})]\, \mathrm d \vec u.
  \end{align*}
\end{theorem}

\section{A reduced density matrix for a prescribed 
  density and paramagnetic current density}
\label{redmat}

Let a density $\rho$ be given. We assume that the density is
non-negative and that it belongs to the intersection of two Lebesgue
spaces,
\begin{equation}
\rho(\mathbf{r}) \geq 0, \quad \text{and} \quad \rho \in L^1(\mathbb{R}^3) \cap L^q(\mathbb{R}^3),
\end{equation}
for some $q > 1$. The latter condition amounts to
\begin{align}
  N := \|\rho\|_1 &= \int \! |\rho(\mathbf{r})| \,\mathrm d\mathbf{r} < +\infty, \\%\intertext{and}
  \|\rho\|_q^q &= \int \! |\rho(\mathbf{r})|^q \,\mathrm d\mathbf{r} < +\infty,
\end{align}
where $N$ is the number of particles in the density $\rho$. (For simplicity, we restrict ourselves to states with integral $N$ but note
the 1-rdm constructions given below are valid also for fractional $N$.) Furthermore, let an arbitrary measurable vector-valued function
$\boldsymbol{\kappa}:\RR^3\rightarrow\RR^3$ be given and let it
prescribe a paramagnetic current density by the relation
\begin{equation}
 \jpvec(\mathbf{r}) = \frac{1}{2} \rho(\mathbf{r}) \boldsymbol{\kappa}(\mathbf{r}).
\end{equation}
We now consider the question:
does there, for every pair of $\rho$ and $\jpvec$ satisfying these
minimal requirements, exist a $D\in\mathcal{D}_{N,1}$ that reproduces $\rho$ and
$\jpvec$? In short, we seek a reduced density matrix $D$ such that
\begin{enumerate}
\item[(a)]\label{itema} $\rho(\mathbf{r}) = \rho_D(\vec{r}) = (\diag D)(\mathbf{r})$
\item[(b)]\label{itemb} $\jpvec(\mathbf{r}) = \Re\vec{c}_D(\vec{r}) =
  -\frac{\mathrm i}{2} (\diag \nabla_1 D)(\vec{r}) + \text{c.c.}$, 
%\item[(b)]\label{itemb} $\jpvec(\mathbf{r}) = \Re\vec{c}_D(\vec{r}) =
%  -\frac{i}{2} (\diag \nabla_1 D)(\vec{r}) + \text{c.c.} = -\frac{\rmi}{2}\nabla_{\vec{r}} D(\mathbf{r},\mathbf{s})) \big|_{\mathbf{s}=\mathbf{r}} + \text{c.c.}$,
\end{enumerate}
assuming that $D$ has a locally finite kinetic energy for 
(b) to be well defined.  We can indeed find such a density matrix $D\in\mathcal{D}_{N,1}$ 
but shall see that
the condition of a locally finite kinetic energy of $D$
implies mild additional conditions on $\rho$ and $\boldsymbol{\kappa}$.

\subsection{Factorized elements $P_\lambda$ and $Q_\lambda$}

Our strategy is to construct explicitly factorized elements
$P_\lambda=G^\dag_\lambda\!\ast G_\lambda$ and $Q_\mu=H^\dag_\mu\ast
H_\mu$ in $\mathcal{D}_{N,1}$ with a locally finite kinetic
energy. Here, $\lambda,\mu>0$ are real parameters that allow some freedom, noting that a
convex combination $D_{\lambda\mu} = (P_\lambda+Q_\mu)/2$ remains in $\mathcal{D}_{N,1}$,
also with a locally finite kinetic energy. The flexibility of having
several independent factorized reduced density matrices $P_\lambda$
and $Q_\mu$ allows the convex combination to reproduce the desired
current.

%As explained in Sec.~\ref{secDIAG}, we require a condition on the
%kinetic energy density in order to ensure that the above diagonals may
%be identified with densities and currents. A sufficient condition is
%that the total canonical kinetic energy is finite. However, the
%canonical kinetic energy is not a physical quantity due to its gauge
%dependence. The physical kinetic energy is the expectation value of
%$\tfrac{1}{2} (-i\nabla + \mathbf{A})^2$. In order to cover the case
%when the canonical kinetic is infinite, but the physical kinetic
%energy is finite due to cancellation from the vector
%potential-dependent contributions, we also will consider a weaker
%condition: the expression (b) remain meaningful if the canonical
%kinetic energy density is locally integrable, i.e.\ is an element of
%$L^1_{\loc}(\mathbb{R}^3)$.

%We note that evaluating a density matrix and its derivative along the
%diagonal is not meaningful per se, since $D$ is only well-defined up
%to a set of measure zero, and since the diagonal is such a set. In
%Appendix A, it is demonstrated how the diagonal is well-defined using
%the limit of a local average along the diagonal. It is also shown,
%that if a factorization $G = P \ast P^\dag$ is given, the diagonal is
%in fact almost everywhere equal to $(h\ast h^\dag)(\vec{r},\vec{r})$,
%and similarly for the derivative.

The two terms are defined by the factorized expressions
  \begin{align}
  \label{eq:PQ1}
    P_{\lambda}(\mathbf{r},\mathbf{s}) & = \sqrt{\rho(\mathbf{r}) \rho(\mathbf{s})} \int_{\RR^3} \! g^\ast(\mathbf{u},\mathbf{r}) g(\mathbf{u},\mathbf{s}) \,\mathrm d\mathbf{u}, \\
  \label{eq:PQ2}
    Q_{\mu}(\mathbf{r},\mathbf{s}) & = \sqrt{\rho(\mathbf{r}) \rho(\mathbf{s})} \int_{\RR^3} \! h^\ast(\mathbf{u},\mathbf{r}) h(\mathbf{u},\mathbf{s}) \,\mathrm d\mathbf{u},
  \end{align}
where 
  \begin{align}
  \label{eq:gh1}
    g(\mathbf{u},\mathbf{v}) & = \frac{\sqrt{8} \lambda^{3/4}}{\pi^{3/4}} \mathrm e^{-\mathrm i\mathbf{v}\cdot\boldsymbol{\kappa}(\mathbf{v})} \mathrm e^{-2\lambda (\mathbf{u}-\mathbf{v})^2},   \\
  \label{eq:gh2}
    h(\mathbf{u},\mathbf{v}) & = \frac{\sqrt{8} \mu^{3/4}}{\pi^{3/4}} \mathrm e^{\mathrm i\mathbf{u}\cdot\boldsymbol{\kappa}(\mathbf{v})} \mathrm e^{-2\mu (\mathbf{u}-\mathbf{v})^2}.
  \end{align}
Clearly, these operators may be written in the form
  \begin{alignat}{2}
  \label{eq:GH1}
P_\lambda &= G_\lambda^\dag \!\ast G_\lambda, &\quad
    G_\lambda(\vec{r},\vec{s}) &= g(\vec{r},\vec{s})\sqrt{\rho(\vec{s})}, \\
  \label{eq:GH2}
Q_\mu &= H_\mu^\dag \ast H_\mu &
    H_\mu(\vec{r},\vec{s}) &= h(\vec{r},\vec{s})\sqrt{\rho(\vec{s})},
  \end{alignat}
It is straightforward to verify that $G_\lambda,H_\lambda\in L^2(\RR^3\times\RR^3)$.

The integration over $\mathbf{u}$ may be performed analytically,
yielding the alternative expressions
  \begin{align}
    P_{\lambda}(\mathbf{r},\mathbf{s}) & = \sqrt{\rho(\mathbf{r}) \rho(\mathbf{s})} \mathrm e^{-\lambda |\mathbf{r}-\mathbf{s}|^2} \mathrm e^{\mathrm i(\mathbf{r}\cdot\boldsymbol{\kappa}(\mathbf{r})-\mathbf{s}\cdot\boldsymbol{\kappa}(\mathbf{s}))}, \label{eq:expanded1}\\
    Q_{\mu}(\mathbf{r},\mathbf{s}) & = \sqrt{\rho(\mathbf{r})
      \rho(\mathbf{s})} \mathrm e^{-\mu|\mathbf{r}-\mathbf{s}|^2} \notag\\ &\quad
    \times \mathrm e^{-\tfrac{\mathrm i}{2} (\mathbf{r}+\mathbf{s})\cdot(\boldsymbol{\kappa}(\mathbf{r})-\boldsymbol{\kappa}(\mathbf{s}))-|\boldsymbol{\kappa}(\mathbf{r})-\boldsymbol{\kappa}(\mathbf{s})|^2/16\mu}.\label{eq:expanded2}
  \end{align}
These operators were found by making the initial ansatz 
$\phi(\mathbf{r}) = \sqrt{\rho(\mathbf{r})}
\mathrm e^{\mathrm i\mathbf{r}\cdot\boldsymbol{\kappa}(\mathbf{r})}$ for an
unnormalized natural orbital. The corresponding paramagnetic current
is then almost correct but contains an extra term that is most easily
canceled if the density matrix contains exponential factors of the
form $\mathrm e^{\mathrm i\mathbf{r}\cdot\boldsymbol{\kappa}(\mathbf{s})}$. Since
the elements of $\mathcal{D}_{N,1}$ and their properties are conveniently
described if an explicit factorization is available (see
Theorem~\ref{thm:summary}), Gaussian kernels are suitable since since
they allow mixed phase factors of the type
$\mathrm e^{\mathrm i\mathbf{r}\cdot\boldsymbol{\kappa}(\mathbf{s})}$ to survive the
integration.

\subsection{The density of $P_\lambda$ and $Q_\lambda$}

We now need to verify that $P_\lambda$ and $Q_\mu$ are elements of
$\mathcal{D}_{N,1}$ by checking points (1)--(5) of Theorem~\ref{thm:DN1properties}.

\begin{theorem}
  Let $\rho\in L^1(\RR^3)\cap L^q(\RR^3)$ for some $q > 1$, $\rho\geq
  0$ a.e., $\|\rho\|_1=N$, and let
  $\lambda,\mu\in\RR$ be such that
  \begin{equation*}
    \label{eqLAMBDAMUCOND}
    \lambda,\mu \geq \frac{2p}{\pi} ( \tfrac{1}{4} N \|\rho\|_q )^{2p/3} 
  \end{equation*}
  where $1/p + 1/q = 1$. 
  Then $P_\lambda$ and $Q_\mu$ in Eqs.~\eqref{eq:PQ1}--\eqref{eq:gh2} are
  elements of $\mathcal{D}_{N,1}$, with
  \begin{equation*}
    \rho_{P_\lambda}(\vec{r}) = \rho_{Q_\mu}(\vec{r}) = \rho(\vec{r})
  \end{equation*}
almost everywhere.
The same is true for any convex
  combination $\theta P_\lambda + (1-\theta)Q_\mu \in
  \mathcal{D}_{N,1}$ with $\theta\in[0,1]$.
\end{theorem}
%\note{What equations do `as defined above' refer to? And what is `(a) above'?} \note{ $\leftarrow$ Eqs.~(32) and (33), or alternatively Eqs.~(38) and (39), define $P_{\lambda}$ and $Q_{\mu}$. The condition labelled `(a)' is just under Eq.~(31). /Erik}
\begin{proof}
  Both operators are Hermitian and positive semi{\-}definite. 
  From the expressions in Eqs.\;\eqref{eq:expanded1} and \eqref{eq:expanded2},
  $(\diag P_\lambda)(\vec{r}) = (\diag Q_\lambda)(\vec{r}) = \rho(\vec{r})$ almost everywhere. 
  It follows that $\tr P_\lambda = \tr Q_\lambda
  = \int \! \rho(\vec{r}) \mathrm d\vec{r} = N$. 

It remains to compute a
  bound on the largest eigenvalues, demonstrating
  point (4) of Theorem~\ref{thm:DN1properties} for the corresponding
  parameter values $\lambda$ and $\mu$.
For an arbitrary normalized orbital,
\begin{align}
  n^2 & \leq \left| \int \!\!\phi^\ast(\mathbf{r}) P_{\lambda}(\mathbf{r},\mathbf{s}) \phi(\mathbf{s}) \,\mathrm d\mathbf{r} \mathrm d\mathbf{s} \right|^2 \notag\\
   & \leq \left( \int \! |\phi(\mathbf{r})| \sqrt{\rho(\mathbf{r}) \rho(\mathbf{s})} \mathrm e^{-\lambda |\mathbf{r}-\mathbf{s}|^2} |\phi(\mathbf{s})| \, \mathrm d\mathbf{r} \, \mathrm d\mathbf{s} \right)^2  \\  
   & = \left( \int \! |\phi(\mathbf{r})| \sqrt{\rho(\mathbf{r})} \left( \int \! \sqrt{\rho(\mathbf{s})} \mathrm e^{-\lambda |\mathbf{r}-\mathbf{s}|^2} |\phi(\mathbf{s})| \mathrm d\mathbf{s} \right) \! \mathrm d\mathbf{r} \right)^2.\notag
\end{align}
Given that $\phi, \sqrt{\rho} \in L^2(\mathbb{R}^3)$, the
Cauchy--Schwarz inequality may be applied twice to give
\begin{align}
  n^2 & \leq \left( \int \! |\phi(\mathbf{r}')|^2 \mathrm d\mathbf{r}' \right) \times \nonumber 
\\ &\quad \times \left( \int \! \rho(\mathbf{r}) \left( \int \!\!  \sqrt{\rho(\mathbf{s})} \mathrm e^{-\lambda |\mathbf{r}-\mathbf{s}|^2} |\phi(\mathbf{s})| \mathrm d\mathbf{s} \right)^2 \! \mathrm d\mathbf{r} \! \right) 
\nonumber \\
   & \leq \int \rho(\mathbf{r}) \left( \int |\phi(\mathbf{s}')|^2 \mathrm d\mathbf{s}' \int \!\rho(\mathbf{s}) \mathrm e^{-2\lambda |\mathbf{r}-\mathbf{s}|^2} \mathrm d\mathbf{s} \right) \! \mathrm d\mathbf{r} 
\nonumber \\
   & \leq \int \! \rho(\mathbf{r}) \left( \sup_{\mathbf{c}} \int \! \rho(\mathbf{s}) \mathrm e^{-2\lambda |\mathbf{c}-\mathbf{s}|^2} \mathrm d\mathbf{s} \right) \mathrm d\mathbf{r} 
\nonumber \\
   & = N \sup_{\mathbf{c}} \int \! \rho(\mathbf{s}) \mathrm e^{-2\lambda |\mathbf{c}-\mathbf{s}|^2} \mathrm d\mathbf{s}
\end{align}
Finally, exploiting the fact that $\rho \in L^q(\mathbb{R}^3)$, the integral over
$\mathbf{s}$ may be bounded by invoking the H\"older inequality,
\begin{equation}
\!\!   n^2  \leq N \|\rho\|_q \sup_{\mathbf{c}} \|\mathrm e^{-2\lambda |\mathbf{c}-\mathbf{s}|^2}\|_p = N \|\rho\|_q \left( \frac{\pi}{2p\lambda} \right)^{3/2p}
\end{equation}
where $1/p + 1/q = 1$.
This bound is independent of the current density. Hence, $P_{\lambda}$
has no eigenvalues greater than two if
\begin{equation}
 \begin{split}
   \label{eqLAMBDACONDITION}
   \lambda \geq \frac{2p}{\pi} \left(\tfrac{1}{4} N \|\rho\|_q \right)^{2p/3}.
 \end{split}
\end{equation}
These steps hold also for $Q_{\mu}$, showing that it has
no eigenvalues greater than 2 when $\mu \geq \frac{2p}{\pi}
(\tfrac{1}{4} N \|\rho\|_q )^{2p/3}$.

Finally, consider a convex combination $D_\theta = \theta P_\lambda +
(1-\theta)Q_\mu$, which belongs to $\mathcal{D}_{N,1}$ since this set convex. Moreover,
$\diag A$ is linear in $A$ since $\diag(A+B)(\vec{r}) =
\diag(A)(\vec{r}) + \diag(B)(\vec{r})$ almost everywhere. Therefore, $\diag D_\theta =
\theta \diag P_\lambda + (1-\theta)\diag Q_\lambda = \rho$ almost
everywhere.
\end{proof}

\ 
\subsection{The canonical kinetic energy of $P_\lambda$ and $Q_\lambda$}
\label{cankin}

We now turn to the question of whether the current $\jpvec$ can be
reproduced by $D$. Indeed, a formal calculation shows that
%The operators $P_{\lambda}$ and $Q_{\mu}$ satisfy (i), (ii) and (a). In fact, they also
%satisfy (iii) and (v) by construction and for suitable parameter values they
%also satisfy (iv). However, they fail to satisfy (b), since
\begin{align}
  -\frac{\mathrm i}{2} \frac{\partial}{\partial r_{\alpha}}
 & P_{\lambda}(\mathbf{r},\mathbf{s})) \big|_{\mathbf{s}=\mathbf{r}} +
  \text{c.c.}  = \rho(\mathbf{r}) \!\left( \!\kappa_{\alpha}(\mathbf{r})
    + \mathbf{r}\cdot \frac{\partial
      \boldsymbol{\kappa}(\mathbf{r})}{\partial r_{\alpha}} \right)
  \notag \\&= 2\jpcomp{\alpha}(\mathbf{r}) + \rho(\mathbf{r})\
  \mathbf{r}\cdot \frac{\partial
    \boldsymbol{\kappa}(\mathbf{r})}{\partial r_{\alpha}} 
\end{align}
and
\begin{align}
  -\frac{\mathrm i}{2} \frac{\partial}{\partial r_{\alpha}}
  Q_{\mu}(\mathbf{r},\mathbf{s})) \big|_{\mathbf{s}=\mathbf{r}} +
  \text{c.c.} & = -\rho(\mathbf{r}) \ \mathbf{r}\cdot\frac{\partial
    \boldsymbol{\kappa}(\mathbf{r})}{\partial r_{\alpha}}. \label{eq:formalcalculation}
\end{align}
Thus, we expect $\jpvec_{D_{\lambda\mu}} =
\frac{1}{2}\jpvec_{P_\lambda}(\vec{r}) +
\frac{1}{2}\jpvec_{Q_\mu}(\vec{r}) = \jpvec(\vec{r})$ to hold almost everywhere. 
To prove this result, it suffices to find conditions on
$\rho$ and $\kappa$ such that $P_\lambda$ and $Q_\mu$ have a locally
finite kinetic energy.
%
%For this reason, we may compute the kinetic energy as
%\begin{equation}
%  \frac{1}{2} \tr(\nabla D \nabla) = \frac{1}{2} \int \nabla_{\mathbf{r}}\cdot \nabla_{\mathbf{s}}
%  D(\mathbf{r},\mathbf{s}) \Big|_{\mathbf{s}=\mathbf{r}} d\mathbf{r} = \int
%  \tau_D(\mathbf{r})d\mathbf{r},
%\end{equation}
%where $\tau_D$ is the kinetic energy density. Using the linearity of the trace,
%\begin{equation}
%  \frac{1}{2} \tr(\nabla_{\mathbf{r}} D_{\lambda\mu} \nabla_{\mathbf{s}}) = \int \tau_P(\math%bf{r})d\mathbf{r} + \int
%  \tau_Q(\mathbf{r})d\mathbf{r},
%\end{equation}

The kinetic energy density of $P_\lambda$ is
\begin{equation}
 \begin{split}
  \tau_P(\mathbf{r}) & =
  \frac{1}{2}\|\nabla_2 G_\lambda(\cdot,\vec{r})\|_{L^2(\RR^3)}^2 = \frac{1}{2} \nabla_{\mathbf{r}}\cdot \nabla_{\mathbf{s}} P_{\lambda}(\mathbf{r},\mathbf{s}) \big|_{\mathbf{s}=\mathbf{r}} \\
   &  = \frac{|\nabla \rho(\mathbf{r})|^2}{8 \rho(\mathbf{r})} + \frac{1}{2} |\nabla (\mathbf{r}\cdot\boldsymbol{\kappa}(\mathbf{r}))|^2 \rho(\mathbf{r}) + \lambda \rho(\mathbf{r}).
 \end{split}
\end{equation}
Here, we have used the fact that the integral in $\|\nabla_2
G(\cdot,\vec{r})\|^2_{L^2(\RR^3)}$ can be performed analytically, so
that the evaluation at $\vec{s}=\vec{r}$ after the second equality is
in fact well-defined.
Similarly, 
\begin{widetext}
\begin{equation}
  \tau_Q(\mathbf{r})  = \frac{1}{2} \nabla_{\mathbf{r}}\cdot \nabla_{\mathbf{s}} Q_{\mu}(\mathbf{r},\mathbf{s}) \big|_{\mathbf{s}=\mathbf{r}}
     = \frac{|\nabla \rho(\mathbf{r})|^2}{8 \rho(\mathbf{r})} \nonumber + 
     \frac{1}{2} \left( \sum_{\alpha=1}^3 \left( \sum_{\beta=1}^3
         r_{\beta} \frac{\partial \kappa_{\beta}(\mathbf{r})}{\partial
           r_{\alpha}} \right)^2 + \tfrac{1}{8\mu} \sum_{\beta=1}^3
       |\nabla \kappa_{\beta}(\mathbf{r})|^2 \right) \rho(\mathbf{r})
     + \mu \rho(\mathbf{r}). 
\end{equation}
The total kinetic energy density becomes 
\begin{equation}
 \begin{split}
  \tau_D(\mathbf{r}) & = \frac{|\nabla \rho(\mathbf{r})|^2}{8 \rho(\mathbf{r})} + \frac{1}{4} \sum_{\alpha=1}^3 \left[ \left(\frac{\partial}{\partial r_{\alpha}} \mathbf{r}\cdot\boldsymbol{\kappa}(\mathbf{r})\right)^2 + \left( \mathbf{r}\cdot \frac{\partial \boldsymbol{\kappa}(\mathbf{r})}{\partial r_{\alpha}} \right)^2 + \tfrac{1}{8\mu} \sum_{\beta=1}^3 \left( \frac{\partial \kappa_{\beta}(\mathbf{r})}{\partial r_{\alpha}} \right)^2 \right] \rho(\mathbf{r}) 
   + \frac{1}{2} (\lambda+\mu) \rho(\mathbf{r}),
 \end{split}
\end{equation}
\end{widetext}
where we have used the fact that the kinetic energy density is linear
in the density matrix.
Using the special form $2|ab|\leq a^2+b^2$ of Young's inequality, with
$a = \kappa_{\alpha}(\mathbf{r})$ and $b=\mathbf{r} \cdot \partial \boldsymbol{\kappa}(\mathbf{r})/\partial r_{\alpha}$, we find that
\begin{align}
  \left| \frac{\partial}{\partial r_{\alpha}}
    \mathbf{r}\cdot\boldsymbol{\kappa}(\mathbf{r}) \right|^2 & =
  \left|\kappa_{\alpha}(\mathbf{r}) + \mathbf{r} \cdot \frac{\partial
      \boldsymbol{\kappa}(\mathbf{r})}{\partial r_{\alpha}} \right|^2 
\notag \\ & \leq 2 |\kappa_{\alpha}(\mathbf{r})|^2 + 2 \left|\mathbf{r} \cdot \frac{\partial \boldsymbol{\kappa}(\mathbf{r})}{\partial r_{\alpha}} \right|^2.
\end{align}
Hence, the canonical kinetic energy density is bounded by 
\begin{widetext}
\begin{equation}
 \begin{split}
  \tau_D(\mathbf{r}) & \leq \frac{|\nabla \rho(\mathbf{r})|^2}{8 \rho(\mathbf{r})} + \frac{1}{2} \left[ |\boldsymbol{\kappa}(\mathbf{r})|^2 + \tfrac{3}{2} \left( \mathbf{r}\cdot \frac{\partial \boldsymbol{\kappa}(\mathbf{r})}{\partial r_{\alpha}} \right)^2 + \tfrac{1}{16\mu} \sum_{\beta=1}^3 \left( \frac{\partial \kappa_{\beta}(\mathbf{r})}{\partial r_{\alpha}} \right)^2 \right] \rho(\mathbf{r}) 
   + \frac{1}{2} (\lambda+\mu) \rho(\mathbf{r}) \\
 & \leq \frac{|\nabla \rho(\mathbf{r})|^2}{8 \rho(\mathbf{r})} + \frac{1}{2} \left[ |\boldsymbol{\kappa}(\mathbf{r})|^2 + (\tfrac{3}{2} r^2 + \tfrac{1}{16\mu}) \sum_{\alpha,\beta=1}^3 \left( \frac{\partial \kappa_{\beta}(\mathbf{r})}{\partial r_{\alpha}} \right)^2  \right] \rho(\mathbf{r}) 
  + \frac{1}{2} (\lambda+\mu) \rho(\mathbf{r}),
 \end{split}
\end{equation}
\end{widetext}
where the
second inequality was obtained by using $|r_{\beta}| \leq
|\mathbf{r}|$.  A finite canonical kinetic energy of $P_\lambda$,
$Q_\mu$ and $D_{\lambda\mu}$
is ensured if
\begin{align}
  \label{eqFINTW}
  T_W[\rho] & = \int \! \frac{|\nabla\rho|^2}{8\rho} \,\mathrm d\mathbf{r} =
  \frac{1}{2}\int \! |\nabla\sqrt{\rho}|^2 \,\mathrm d\vec{r} < \infty, \\
  \label{eqFINTp}
  T_{\mathrm{p}}[\rho,\jpvec] &  = \int \! \frac{|\jpvec|^2}{2\rho} \,\mathrm d\mathbf{r} = \frac{1}{8} \int \! \rho \kappa^2 \,\mathrm d\mathbf{r} < \infty, \\
  \label{eqFINTab}
  T_{\alpha\beta}[\rho,\jpvec] & = \int \! (1+r^2) \rho \left( \frac{\partial \kappa_{\beta}(\mathbf{r})}{\partial r_{\alpha}} \right)^2 \mathrm d\mathbf{r} < \infty.
\end{align}
We remark that, by the definition of the vorticity, $\boldsymbol{\nu} = \nabla\times\rho^{-1} \jpvec =
\frac{1}{2} \nabla\times\boldsymbol{\kappa}$, a consequence of the last condition
is that
\begin{equation}
  \int \! (1+r^2) \rho \nu^2 \mathrm d\mathbf{r} < \infty.
\end{equation}
We have thus proved the following result:
\begin{theorem}\label{thm:globalrep}
  Let $\rho$ and $\boldsymbol{\kappa}$ be given such that $\rho\geq
  0$, $\sqrt{\rho} \in H^1(\mathbb{R}^3)$, $\rho \kappa_{\alpha}^2 \in
  L^1(\mathbb{R}^3)$, $(1+r^2) \rho (\partial\kappa_{\beta}/\partial
  r_{\alpha})^2 \in L^1(\mathbb{R}^3)$ for all Cartesian components
  $\alpha,\beta\in\{1,2,3\}$. Then there exist real constants
  $\lambda,\mu\geq 0$ and a 1-rdm $D$ with density $\rho$ and current
  $\jpvec = \frac{1}{2} \rho \boldsymbol{\kappa}$ such that the
  canonical kinetic energy is bounded by
  \begin{align*}
    \frac{1}{2} &\tr(\boldsymbol \nabla_1  \cdot\boldsymbol \nabla_2 D ) \leq T_W[\rho] +  4
    T_{\mathrm{p}}\left[\rho, \tfrac{1}{2} \rho\kappa\right] + \frac{1}{2}
    (\mu+\nu)N \notag\\ 
    & \quad + \int \! \!\rho(\mathbf{r}) \left(\tfrac{3}{4} r^2 + \tfrac{1}{32\mu}\right) \!\sum_{\alpha,\beta=1}^3 \!\left( \frac{\partial \kappa_{\beta}(\mathbf{r})}{\partial r_{\alpha}} \right)^2 \!\mathrm d\mathbf{r}.
  \end{align*}
\end{theorem}
Note that, by a Sobolev inequality, $\sqrt{\rho}\in H^1(\RR^3)$ implies that $\rho\in L^q(\RR^3)$ for all $q\in [1,3]$.

\subsection{Lifting the global integrability condition}

Theorem~\ref{thm:globalrep} may be strengthened by replacing the
global integrability conditions on the total kinetic energy by local
integrability conditions, replacing integrals $\RR^3$ by integrals
over arbitrary compact sets $K\subset\RR^3$. A larger class of $\rho$
and $\boldsymbol{\kappa}$ are then seen to be reproducible, albeit
with merely a locally finite kinetic energy.
\begin{theorem}
  Let $\rho$ and $\boldsymbol{\kappa}$ be given such that $\rho\geq
  0$, $\rho \in L^1(\mathbb{R}^3)\cap L^q(\RR^3)$, $q>1$, $\rho^{-1}
  |\boldsymbol \nabla\rho|^2 \in L^1_{\loc}(\mathbb{R}^3)$, $\rho
  \kappa_{\alpha}^2 \in L^1_{\loc}(\mathbb{R}^3)$, $(1+r^2) \rho
  (\partial\kappa_{\beta}/\partial r_{\alpha})^2 \in
  L^1_{\loc}(\mathbb{R}^3)$ for all Cartesian components
  $\alpha,\beta\in\{1,2,3\}$. Then there exist a 1-rdm $D$ with
  density $\rho$ and current $\jpvec = \frac{1}{2} \rho
  \boldsymbol{\kappa}$.
\end{theorem}

%We have proved the following result:
%\begin{theorem}
%  Suppose $\rho\geq 0$, $\rho^{1/2} \in H^1(\mathbb{R}^3)$. Suppose
%  $\kappa_{\alpha}, \partial\kappa_{\beta}/\partial r_{\alpha} \in
%  L^2(\rho(\vec{r}) d\vec{r})$ and $\partial\kappa_{\beta}/\partial
%  r_{\alpha} \in L^2(r^2\rho(\vec{r})d\vec{r})$ for all Cartesian components
%  $\alpha,\beta\in\{1,2,3\}$. Then there exists real constants
%  $\lambda,\mu\geq 0$ and a 1-rdm $D$ with density $\rho$ and current
%  $\jpvec = \frac{1}{2} \rho \boldsymbol{\kappa}$ such that
%  \begin{equation}
%    \frac{1}{2} \tr(\nabla D \nabla) \leq T_W[\rho] + 4
%    T_{\mathrm{p}}[\rho, \tfrac{1}{2} \rho\kappa] + \frac{1}{2} (\mu+\nu)N + \int \rho(\mathbf{r}) (\tfrac{3}{4} r^2 + \tfrac{1}{32\mu}) \sum_{\alpha,\beta=1}^3 \left( \frac{\partial \kappa_{\beta}(\mathbf{r})}{\partial r_{\alpha}} \right)^2 d\mathbf{r}.
%  \end{equation}
%\end{theorem}

\section{Discussion}
\label{discussion}

We have provided an explicit construction of a 1-rdm that reproduces a
prescribed density and paramagnetic current density. This type of
$N$-representability problem arises in Kohn--Sham CDFT, as it is known
that not all current densities can be represented by a single
Kohn--Sham orbital. Lieb and Schrader have recently proved
\cite{LIEB_SCHRADER_2013}, under some additional assumptions, that
there also exist current densities that cannot be represented by two
Kohn--Sham orbitals. The question is open for three orbitals. For four or more
orbitals, Lieb and Schrader provide an explicit Slater determinant
that reproduces any density and paramagnetic current that satisfy
mild regularity conditions. Our results are complementary in that we
establish that an \emph{extended} Kohn--Sham approach, where
fractional occupation numbers are allowed even if there is an integral
total number of electrons, is flexible enough to represent any density
and paramagnetic current density, under minimal regularity assumptions
(finite $T_{\mathrm{W}}$, $T_{\mathrm{p}}$, and $T_{\alpha\beta}$).

The generalization from finite total canonical kinetic energy to
finite local canonical kinetic energy is of some value in light of
gauge freedom. The kinetic energy $T_{\mathrm{p}}[\rho,\jpvec]$ is not
gauge invariant; on the contrary, it can made to become infinite by
applying a gauge transformation $\jpvec \mapsto \jpvec + \rho \nabla \chi$
with a rapidly growing gauge function $\chi$. Our results establish
that such gauge transformations do not affect $N$-representability, as
long as $\chi$ exhibits some minimal regularity.

The explicit constructions of density matrices can be used to provide
orbital-free upper bounds on the canonical kinetic energy
$T_s[\rho,\jpvec]$ for an extended Kohn--Sham formalism. Combining the
above results with the standard lower bound $T_W + T_{\mathrm{p}}$ on
the kinetic energy, we get the following orbital-free bounds on the
extended Kohn--Sham kinetic energy,
\begin{align}
  T_W + T_{\mathrm{p}} &\leq \bar{T}[\rho,\jpvec] \notag \\
  & \leq T_W + 4 T_{\mathrm{p}} + (\lambda+\mu) N \notag \\
 & \quad + \int \! \rho(\mathbf{r}) \left(\tfrac{3}{2} r^2 + \tfrac{1}{16\mu}\right) \left|\nabla_{\alpha} \kappa_{\beta}(\mathbf{r})\right|^2 \mathrm d\mathbf{r}.
\end{align}

Noting that  several authors, following Vignale and
Rasolt~\cite{VIGNALE_PRB37_10685}, have discussed CDFT formulations in
terms of spin-resolved densities
$(\rho_{\uparrow},\rho_{\downarrow},\jpvecix{\uparrow},\jpvecix{\downarrow})$,
we also remark that our 1-rdm construction is easily modified for
spin-resolved 1-rdms $D^{\uparrow\uparrow}$ and
$D^{\downarrow\downarrow}$. The eigenvectors then correspond to natural
spin-orbitals with eigenvalues bounded by one rather than by two as in the
case of natural spatial orbitals. The modifications to the above
presentation are trivial---condition \ref{item:prop4} in
Theorem~\ref{thm:DN1properties} becomes that no occupation is larger
than 1, and the factors $\tfrac{1}{4}$ consequently disappear from
Eqs.~\eqref{eqLAMBDAMUCOND} and \eqref{eqLAMBDACONDITION}.

\acknowledgments

The authors would like to thank E.H.~Lieb and R.~Schrader for giving
one of us (S.~Kvaal) early access to their manuscript of
Ref.~\cite{LIEB_SCHRADER_2013} and for interesting discussions, which
spurred the completion of the present work.

This work was supported by the Norwegian Research Council through the
CoE Centre for Theoretical and Computational Chemistry (CTCC) Grant
No.\ 179568/V30 and the Grant No.\ 171185/V30 and through the European
Research Council under the European Union Seventh Framework Program
through the Advanced Grant ABACUS, ERC Grant Agreement No.\ 267683.

\appendix
\section{The regular representation of trace-class operators}

The density matrix $D(\vec{r},\vec{s})$ is an element of
$L^2(\RR^3\times \RR^3)$ and also the kernel of a trace-class
operator over $L^2(\RR^3)$. As such, it is not pointwise defined
everywhere. At the same time, we wish to make sense of ``the
diagonal $D(\vec{r},\vec{r})$'' in order to define the density in an unambiguous manner.

Brislawn \cite{BRISLAWN_1988} has presented a thorough study of
trace-class operators and their kernels. The basic tools are found in this
reference, but we restate some results for a self-contained
treatment. We begin by clarifying some points concerning Lebesgue spaces that are
often glossed over but are important here.

\subsection{Lebesgue spaces}

Let $X\subset\RR^n$ be an open set.
The $L^p(X)$ norm of a measurable function $f : X\rightarrow
\CC$ is defined by
\begin{equation}
  \|f\|_p := \left( \int_X |f(x)|^p \,\mathrm dx \right)^{1/p}.
\end{equation}
The vector space $\mathcal{L}^p(X)$ consists of all functions $f$ such that
$\|f\|_p<+\infty$. The space $\mathcal{L}^p(X)$ is not a normed space, since
$\|f\|_p = 0$ if and only if $f(x) = 0$ for almost all $x\in
X$ (rather than for all $x \in X$). On the other hand, the set $L^p(X)$ consisting of all \emph{equivalence classes} $[f] =
\{ g\in\mathcal{L}^p(X) : \|f-g\|_p = 0 \}$ is a normed space. It is customary to speak
of a function $f$ as an element of $L^p(X)$ even though, strictly speaking, it is a
\emph{representative} of $[f]\in L^p(X)$.

This distinction between $f$ and $[f]$ is not merely academic: two pointwise defined wave functions $\Psi$ and
$\Phi$ describe the same physical state if and only if $\|\Psi
-\Phi\|_2 = 0$. Thus, $[\Psi]\in L^2(\RR^{3N})$ is the
wave function. Similarly, a reduced density matrix $D\in
L^2(\RR^3\times\RR^3)$ is not defined pointwise: its formal diagonal
$D(\vec{r},\vec{r})$ may therefore be redefined without changing the
physics. If $[\Psi] = [\Phi]$, then $D_\Psi = D_\Phi$ almost everywhere, but if $D$ is
\emph{given} there is no \emph{a priori} way to know how the pointwise
values $D(\vec{r},\vec{s})$ are affected by modifying the wave function on a set
of zero measure.

\subsection{Locally integrable functions}
\label{secLOCINT}

A function $f\in L^p_\loc(\RR^n)$ if and only if $f\in L^p(K)$ for every
compact measurable $K\subset\RR^n$. We furthermore have $L^{q}_\loc\subset
L^p_\loc$ for $q\geq p$, and 
\begin{equation}
L^p(\RR^n)\subset L^p_\loc(\RR^n)
\subset L^1_\loc(\RR^n).\label{eq:localspaces}
\end{equation}
Clearly, $L^1_\loc$ is a large class of functions
and functions in $L^1_\loc$ are said to be
``locally integrable''. 
We also need a
slightly more general notion of local integrability as follows:
\begin{definition}
  Let $X\subset\RR^n$, $Y\subset\RR^m$ be open sets. The set
  $L^p(X_\loc\times Y)$ is the set of (equivalence classes of) all
  measurable functions $u : X\times Y\rightarrow \CC$ such that for
  all compact measurable $K\subset X$, $u\in L^p(K\times Y)$. A
  similar definition is made for arbitrary products and positions of
  the subscript ``$\loc$''.  In particular, $L^p_\loc(X) =
  L^p(X_\loc)$.
\end{definition}

\subsection{The regular representation}

The goal of this section is to establish a unique representative
$\tilde{f}$ of $[f]\in L^1_\loc$, called the regular representative of
$f$. This representative will aid in defining the diagonal of $D\in
\mathcal{D}_{N,1}$. The first step is to introduce the local averaging
operator $A_\epsilon$:

\begin{definition}[Local averaging operator $A_\epsilon$]
  Let $f\in L^1_\loc(\RR^n)$ and $\epsilon>0$. 
  For a box $C_\epsilon = [-\epsilon,\epsilon]^n$ of Lebesgue measure $|C_\epsilon| = (2\epsilon)^n$,
  the (linear) local averaging operator $A_\epsilon : L^1_\loc \rightarrow L^1_\loc$ is defined by 
  \begin{equation}
    A_\epsilon f(x) := \frac{1}{|C_\epsilon|} \int_{C_\epsilon} \! f(x+y) \,\mathrm dy. \label{Aeps}
  \end{equation}
\end{definition}

Since $C_\epsilon$ is compact, $A_\epsilon f(x)$ is 
everywhere finite and is 
independent of the particular representative $f$ of $[f]$ that appears in  
the integrand. It can be shown that $A_\epsilon f(x)$ is continuous both in $x$ and
in $\epsilon>0$ \cite{STEIN_1970}.
We are here interested in the limit $\epsilon\rightarrow 0$ and therefore
invoke the Lebesgue differentiation theorem:
\begin{theorem}[Lebesgue differentiation theorem]\label{thm:lebesguediff}
  Let $f\in L^1_\loc(\RR^n)$. Then for almost all $x\in \RR^n$,
  \begin{equation}
    \lim_{\epsilon\rightarrow 0} A_\epsilon f(x) = f(x). \label{eq:lebesguediff}
  \end{equation}
\end{theorem}
\begin{proof}
  See Ref.~\cite{STEIN_1970}
\end{proof}
Since $A_\epsilon f(x)$ is independent of the particular $f\in[f]$, 
this limit determines a unique representative:
\begin{definition}[Regular representative]
  The regular representative $\tilde{f}$ of 
  $f \in L^1_\loc(\RR^n)$ is defined by
  \begin{equation}
    \tilde{f}(x) := \lim_{\epsilon\rightarrow 0} A_\epsilon f(x)
  \end{equation}
whenever the limit in Eq.\;\eqref{eq:lebesguediff} exists.
\end{definition}
Since $\tilde{f}(x)=f(x)$ almost everywhere, $\tilde{f}$ and $f$ represent 
the same element $[f] \in L^1_\loc$. 
Moreover, it is easy to see that $\tilde{f}$ is independent of the
starting representative $f$ and that the set of zero measure
(where $\tilde{f}$ is undefined) is uniquely given by $[f]\in
L^1_\loc$. Intuitively, $\tilde{f}$ is more regular than
$f$, ``smoothing out'' unnecessary discontinuities, and so on. 

Related to the regular representative is the Hardy--Littlewood maximal
function and associated inequality:
\begin{definition}[Hardy--Littlewood maximal function]
  For $f \in L^1_\loc(\RR^n)$ and $C_\epsilon =
  [-\epsilon,\epsilon]^n$ of Lebesgue measure $|C_\epsilon| =
  (2\epsilon)^n$, the Hardy--Littlewood maximal function $Mf$ is
  defined by
  \begin{equation}
    Mf(x) := \sup_{\epsilon>0}
    \frac{1}{|C_\epsilon|} \int_{C_\epsilon} \! |f(x+y)|\, \mathrm d y.
  \end{equation}
\end{definition}
%A basic but important result is the following theorem:
The following theorem is also called the Maximal Theorem:
\begin{theorem}[Hardy--Littlewood maximal inequality]
  If $f\in L^p(\RR^n)$, then $Mf(x)$ is finite almost everywhere. 
  Moreover, there exists a constant $C_p$ (independent of $f$ and
  $n$) such that
  \[ \|Mf\|_p \leq C_p \|f\|_p. \]
\end{theorem}
\begin{proof}
  See Ref.~\cite{STEIN_1970}.
\end{proof}

Since $|A_\epsilon f(x)| \leq Mf(x)$ for all $x$, we obtain as a corollary that
$A_\epsilon$ is a bounded linear operator from $L^p$ to $L^p$. 
Using this fact, it is straightforward to show that $A_\epsilon$ not
only smoothes $f$, but also the mode of convergence:
\begin{lemma}
  Suppose $f_n\rightarrow f$ in $L^p(X)$. For all $\epsilon>0$,
  $A_\epsilon f_n\rightarrow A_\epsilon f$ uniformly (i.e., in $L^\infty(X)$).
\end{lemma}
\begin{proof}
  We show that, for every $\epsilon>0$, there exists a constant
  $K(\epsilon)$ such that, for all $f\in L^p$,
  \[ \|A_\epsilon f\|_\infty \leq K(\epsilon)\|f\|_p.\]
  We have
  \begin{align*}
    |A_\epsilon f(x)| \leq \frac{1}{|C_\epsilon|} \|f\|_{L^1(x + C_\epsilon)}.
  \end{align*}
  Since $C_\epsilon$ is bounded in $\RR^n$,
  \begin{equation}
%    \int_{x + C_\epsilon} \! |g(x)| \,\mathrm d x= 
\int_{x + C_\epsilon} \!\!\!\!\!\!1\times |g(x)| \,\mathrm d x \leq
    |C_\epsilon|^{1/q}\|g\|_{L^p(x + C_\epsilon)} 
  \end{equation}
  where $1/q + 1/p = 1$. Thus,
  \[   |A_\epsilon f(x)| \leq |C_\epsilon|^{1/q-1}
  \|f\|_{L^p(\RR^n)}, \]
  independent of $x$.
\end{proof}

%\subsection{The density from the density matrix in the regular representation}
\subsection{The diagonal of a factorized kernel} 

Based on our intuition, we may now hypothesize that, given an arbitrary reduced
density matrix $D(\vec{r},\vec{s})\in\mathcal{D}_{N,1}$, the diagonal of $\tilde{D}$ is
the proper definition of the density:
\begin{equation}
  \rho(\vec{r}) = \tilde{D}(\vec{r},\vec{r}).
\end{equation}
This is indeed true, as we shall show.
To this end, a slight reformulation and generalization of Theorem 3.5 in
\cite{BRISLAWN_1988} is useful for us. The reformulation states that, if an
operator kernel is factorized, then the diagonal of the regular
representative is given by the diagonal of the factorization,
almost everywhere. The proof carries over with only
trivial modifications, but since it is important, we rephrase it
here.

\begin{theorem}[Diagonal of factorization]\label{thm:factorization}
  Let $(X,\mathrm dx)$ and $(Y,\mathrm dy)$ be open subsets of Euclidean
  spaces equipped with the standard Lebesgue measures.
  For $P\in L^2(X\times Y)$ and $Q\in L^2(Y\times X)$,
let $C : X\times X\rightarrow \CC$ be given by
\begin{equation}
  C(x,x') = (P\ast Q)(x,x') = \int_Y \!\! P(x,y)Q(y,x') \, \mathrm dy.
\end{equation}
Then $C\in L^2(X\times X)$ (a pointwise representative) and 
\begin{equation}
  \tilde{C}(x,x) = C(x,x)
\end{equation}
for almost all $x\in X$.
Moreover, the map $x \mapsto C(x,x) = (P\ast Q)(x,x)$ 
belongs to $L^1(X)$.
\end{theorem}
\begin{proof}
%  The condition $P\in
%  L^2(X\times Y)$ implies that $P(\cdot,y)\in L^2(X)$ for almost every $y\in Y$ and likewise that
%  $P(x,\cdot)\in L^2(Y)$ for almost every $x\in X$.
%  Mutatis mutandi, the same holds for $Q\in L^2(Y\times X)$.

  We now demonstrate that $C\in L^2(X\times X)$.
  For almost all $x\in X$ and for almost all $x'\in
  X$, it holds that  $P(x,\cdot),Q(\cdot,x') \in L^2(Y)$. From the Cauchy--Schwarz inequality, we obtain 
  \begin{align}
    \left|C(x,x')\right| &\leq \int \!
    \left|P(x,y)\right|\left|Q(y,x')\right|\,\mathrm d y \nonumber \\
    &\leq \|P(x,\cdot)\|_{L^2(Y)}\|Q(\cdot,x')\|_{L^2(Y)} < +\infty
  \end{align}
  for almost all $x$ and almost all $x'$ and hence also for almost all
  $(x,x^\prime)\in X\times X$. 
  Squaring and integrating, we obtain $\|C\|_{L^2(X\times X)}^2 \leq
  \|P\|^2_{L^2(X\times Y)} \|Q\|^2_{L^2(Y\times X)}<+\infty$.

  Next, we demonstrate that the diagonal is in
  $L^1(X\times X)$.
  For $\epsilon>0$, let $A_{\epsilon,i} P(x,y)$ be the averaging
  operator acting on the $i$th argument and let $M_i P(x,y)$ be the
  maximal operator acting on the $i$th argument. For almost all
  $x,x',y$, we then obtain
  \begin{equation}
    |A_{\epsilon,1} P(x,y) A_{\epsilon,2} Q(y,x')| \leq M_1 P(x,y) M_2
    Q(y,x').
    \label{eq:bound1}
  \end{equation}
  By the Cauchy--Schwarz inequality, we obtain
  \begin{align}
    \int \! & \left|M_1 P(x,y)  M_2 Q(y,x')\right|^2 \mathrm d y\leq \notag \\ 
     & \left(\int \! \left|M_1 P(x,y)\right|^2\mathrm d y\!\right)\left(\int \! \left|M_2 Q(y,x')\right|^2\mathrm d y\!\right)
    \label{eq:bound2}
  \end{align}
  where both factors on the right-hand side are finite by the
  maximal theorem, for almost all $x$ and almost all $x'$.
  These bounds justify the use of Fubini's theorem to write
  \begin{align}
    A_\epsilon C(x,x') &= \frac{1}{|C_\epsilon|^2}
    \int_{C_\epsilon\times C_\epsilon\times Y} 
\!\!\!\!\!\!\!\!\!\!
\!\!\!\!\!\!\!\!\!\!\!
    P(x+t,y)Q(y,x'+t') \,\mathrm dt \mathrm dt' \mathrm dy
\notag\\&= \int_{Y} \! A_{\epsilon,1} P(x,y) \,
    A_{\epsilon,2} Q(y,x') \,\mathrm d y ,
    \label{eq:fubini}
  \end{align}
  which holds for almost every $x$ and $x'$.

  We now observe that, for each factor on the right-hand side,
  \begin{alignat}{2}
    \lim_{\epsilon\rightarrow 0} A_{\epsilon,1} P(x,y) &= P(x,y) &\quad
    &\mbox{a.a.~$x\in X$} \\
    \lim_{\epsilon\rightarrow 0} A_{\epsilon,2} Q(y,x') &= Q(y,x') &
    &\mbox{a.a.~$x'\in X$} .
  \end{alignat}
  The dominated convergence theorem together with the bounds in Eqs.\;\eqref{eq:bound1} and \eqref{eq:bound2} now
  imply that we can take the limit in Eq.\;\eqref{eq:fubini} to get
  \begin{equation}
    \tilde{C}(x,x) = \lim_{\epsilon\rightarrow 0} A_\epsilon C(x,x) = \int_Y \! P(x,y)Q(y,x) \,\mathrm d y
  \end{equation}
  for almost all $x$. We have
  \begin{equation}
    \int_X \! (P*Q)(x,x) \, \mathrm d x= \braket{\hat{Q},P}_{L^2(X\times Y)},
  \end{equation}
  with $\hat{Q}(x,y) = Q^\ast(y,x)$. Being an inner product on $L^2$,
  this expression is finite, completing the proof. 
\end{proof}

Remark 1: Although the diagonal of $P*Q$ is in $L^1$, we cannot conclude that
$P*Q$ is trace class---see Ref.\,\cite{BRISLAWN_1988} for a counterexample. On the other hand, if $X = Y$ in
Theorem~\ref{thm:factorization}, then $P*Q$ is by definition trace class 
and it is also true that $\tr P*Q = \int_X(\diag P*Q)(x)\,\mathrm d x$.

Remark 2: $C = P*Q$ is the kernel of a Hilbert--Schmidt operator over
$L^2(X)$. We see that it is meaningful to define the diagonal $\diag C$ of
any Hilbert--Schmidt operator on an explicitly factorized form from the expression
\begin{equation}
  [\diag P*Q](x) := (P*Q)(x,x),
\end{equation}
and the theorem states that this function belongs to $L^1(X)$, independent of the factorization.

Remark 3: As a corollary, if $P\in L^2(\RR^n_\loc\times \RR^m)$, $Q\in L^2(\RR^m\times
\RR^n_\loc)$, then $P\ast Q \in L^1_\loc(\RR^n)$.

Remark 4: If $P(x,y) = Q^\ast(y,x)$, then $P\ast Q$ is positive
semidefinite. Since $\diag P\ast Q$ is integrable, it follows from a
theorem in Ref.~\cite{BRISLAWN_1988} that $P\ast Q$ is trace class over $L^2(X)$.

\section{Some proofs from Section~\ref{secDIAG}}
\subsection{Proof of Theorem\;\ref{thm:densities}}

%We restate Theorem~\ref{thm:densities} and prove it:
%\begin{theorem}\label{thm:densities-app}\note{Copy formulations from
%    Sec. 2 for final version.}
%  Let $D\in \mathcal{D}_{N,1}$, and suppose that for a $G\in
%  L^2(\RR^3\times\RR^3)$,
%  \begin{equation}
%    D(\vec{r},\vec{s}) = (G^\dag * G)(\vec{r},\vec{s}),
%  \end{equation}
%  i.e.,
%  \begin{equation}
%    D(\vec{r},\vec{s}) = \int_{\RR^3} G(\vec{u},\vec{r})^*
%    G(\vec{u},\vec{s}) d\vec{u} \quad \text{a.e.} [d\vec{r}d\vec{s}]
%  \end{equation}
%  Then, for every $\Gamma\in \mathcal{D}_N$  with $D = D_\Gamma$
%  a.e. (that is, $\Gamma\mapsto D$),
%  \begin{equation}
%    \rho_\Gamma(\vec{r}) = (G^\dag*G)(\vec{r},\vec{r}) \quad
%    \text{a.e.} [d\vec{r}],
%  \end{equation}
%  i.e.,
%  \begin{equation}
%    \rho_\Gamma(\vec{r}) = \int_{\RR^3} G(\vec{u},\vec{r})^*
%    G(\vec{u},\vec{r}) d\vec{u} \quad \text{a.e.} [d\vec{r}].
%  \end{equation}
%\end{theorem}
%

For this proof, we use Theorem~\ref{thm:factorization} in Appendix~A.

\begin{proof}
  Let $\Gamma\in\mathcal{D}_N$ be given. Assume that
  that
  $D_\Gamma(\vec{r},\vec{s})  =
  D(\vec{r},\vec{s})$ almost everywhere in $\RR^3\times\RR^3$. It
  follows that $\tilde{D}_{\Gamma}(\vec{r},\vec{r}) =
  \tilde{D}(\vec{r},\vec{r}) = (G^\dag*G)(\vec{r},\vec{r})$
  for almost all $\vec{r}$, since the regular representative is
  unique, and by using Theorem~\ref{thm:factorization}, with
  $P(\vec{r},\vec{s}) = G(\vec{s},\vec{r})^*$ and
  $Q(\vec{r},\vec{s})=G(\vec{s},\vec{r})$ ($X = Y = \RR^3$).

  We need to show that $\rho_\Gamma(\vec{r}) =
  (G^\dag*G)(\vec{r},\vec{r})$ for almost all $\vec{r}$. Assume that
  $\Gamma = \ket{\Psi}\bra{\Psi}$. Now, $\rho_\Gamma(\vec{r}) =
  \rho_\Psi(\vec{r}) = D_\Gamma(\vec{r},\vec{r})$ for almost every
  $\vec{r}$, by definition of $\rho_\Psi(\vec{r})$. Applying
  Theorem~\ref{thm:factorization} to $P(\vec{r},\vec{r}_{2:N}) =
  \Psi(\vec{r},\vec{r}_{2:N})$ and $Q(\vec{r}_{2:N},\vec{r}) =
  \Psi(\vec{r},\vec{r}_{2:N})^*$, $X = \RR^3$ and $Y = \RR^{3N-3}$, we
  see that $\rho_\Gamma(\vec{r}) = \tilde{D}_\Gamma(\vec{r},\vec{r}) =
  (G^\dag*G)(\vec{r},\vec{r})$.

  We invite the reader to fill in the details when $\Gamma$ is a general mixed state.
\end{proof}

\subsection{Proof for Theorem\;\ref{thm:lockin}} 

%We now turn to the momentum density and paramagnetic current. As
%described in Section~\ref{secDIAG}, we only assign a momentum density
%to $D\in\mathcal{D}_{N,1}$ that have locally finite kinetic energy.
%We restate Lemma~\ref{lemma:lockin} and prove it:
%\begin{lemma}\label{lemma:lockin-app}\note{Copy formulations from
%    Sec. 2 for final version.}
%  Let $D\in\mathcal{D}_{N,1}$.   The following are equivalent:
%  \begin{enumerate}
%    \item
%      $D$ has locally finite kinetic energy.
%    \item
%      There exists a factorization $D = G^\dag \ast G$ with
%      $\nabla_2 G\in L^2(\RR^3\times K)$ for every compact
%      $K\subset\RR^3$. 
%    \item
%      Any $\Gamma\in\mathcal{D}_{N,1}$ with $\Gamma\mapsto D$ has
%      locally finite kinetic energy
%    \end{enumerate}
%\end{lemma}
\begin{proof}

  $2 \Rightarrow 1$:
  Let a $G\in L^2(\RR^3\times\RR^3)$ be given such that $\nabla_2 G \in
  L^2(\RR^3\times\RR^3_\loc)$. Then, for every compact $K\subset
  \RR^3$, 
  \begin{align}
    T_\alpha(\vec{r},\vec{s}) &:= \frac{1}{2} [\partial_{2,\alpha} G]^\dag *
    [\partial_{2,\alpha}G](\vec{r},\vec{s}) \notag \\ &= \frac{1}{2} \int
    d\vec{u} \partial_{2,\alpha}G(\vec{u},\vec{r})^*\partial_{2,\alpha} G(\vec{u},\vec{s})
  \end{align}
  is in $L^2(K\times K)$ by
  Theorem~\ref{thm:factorization}. $T_\alpha$ is positive
  semidefinite, so by Remark 4 after Theorem~\ref{thm:factorization},
  $T_\alpha$ is trace class over $L^2(K)$.

  By the definition of the weak derivative and Fubini's Theorem, we
  easily verify that in fact $T_\alpha
  = \frac{1}{2} \partial_{1,\alpha}\partial_{2,\alpha} D$ almost everywhere. Thus
  $\nabla_1\cdot\nabla_2 D$ is trace-class, and $D$ has locally finite
  kinetic energy.

  $1 \Rightarrow 2$:

  Since $D\in\mathcal{D}_{N,1}$ there exists a spectral decomposition
  \begin{equation}
    B(\vec{r},\vec{s}) = \sum_k \lambda_k \phi_k(\vec{r})\phi_k(\vec{s})^*,
  \end{equation}
  where $\{\phi_k\}\subset L^2(\RR^3)$ is a complete, orthonormal set,
  and where $0\leq\lambda_k\leq 2$ such that $\sum_k \lambda_k =
  N$. Of course $B(\vec{r},\vec{s}) = D(\vec{r},\vec{s})$ almost
  everywhere, but they may be pointwise different. 

  Let $K\subset\RR^3$ be compact. Restricted to $K\times K$,
  $\nabla_1\cdot\nabla_2 D = \nabla_1\cdot \nabla_2 B$ (a.e.) is
  trace-class, and we compute
  \begin{equation}
    \nabla_1\cdot\nabla_2 B(\vec{r},\vec{s}) = \sum_k \lambda_k
    \nabla\phi_k(\vec{r})\cdot\nabla\phi_k(\vec{s})^*\quad \text{a.e.}
  \end{equation}
  Let $A_k(\vec{r},\vec{s}) = 
  \nabla\phi_k(\vec{r})\cdot\nabla\phi_k(\vec{s})^*$. By assumption,
  \begin{equation}
    \tr (\nabla_1\cdot\nabla_2 B) = \sum_k \lambda_k \tr A_k = \sum_k
    \lambda_k \|\nabla\phi_k\|^2_{L^2(K)} < +\infty,
  \end{equation}
  implying that $\nabla\phi_k\in L^2(K)$ for every $K$, hence
  $\nabla\phi_k\in L^2_\loc(\RR^3)$.

  Let $G$ be given by
  \begin{equation}
    G(\vec{r},\vec{s}) = \sum_k \lambda_k^{1/2} \phi_k(\vec{r})\phi_k(\vec{s})^*.
  \end{equation}
  Clearly, $G \in L^2(\RR^3\times\RR^3)$ and $D = G^\dag * G$.
  Moreover,
  \begin{equation}
    \nabla_2 G(\vec{r},\vec{s}) = \sum_k \lambda_k^{1/2} \phi_k(\vec{r})\nabla\phi_k(\vec{s})^*.
  \end{equation}
  Computing the $L^2(\RR^3\times K)$ norm,
  \begin{align}
    \|\nabla_2 G\|^2 &= \sum_{k\ell} \lambda_k^{1/2}\lambda_\ell^{1/2}
    \braket{\phi_\ell,\phi_k}_{L^2(\RR^3)}\braket{\nabla\phi_k,\nabla\phi_\ell}_{L^2(K)}
    \notag \\ &=
    \sum_k \lambda_k \|\nabla\phi_k\|^2_{L^2(K)}.
  \end{align}

  $3\Leftrightarrow 1$:

  Let $\Gamma$ be such that $D_\Gamma = D$ a.e. We have,
  \begin{equation}
    D_\Gamma(\vec{r},\vec{s}) = \sum_i p_i \int d\vec{r}_{2:N}
    \Psi_i(\vec{r},\vec{r}_{2:N})\Psi_i(\vec{s},\vec{r}_{2:N})^*. 
  \end{equation}
  Furthermore,
  \begin{align}
    T_\alpha(\vec{r}, & \vec{s})
     := \frac{1}{2} \partial_{1,\alpha}\partial_{2,\alpha}D(\vec{r},\vec{s})
    \notag \\ & =
    \frac{1}{2} \sum_i p_i \int d\vec{r}_{2:N} 
    \partial_{1,\alpha}\Psi_i(\vec{r},\vec{r}_{2:N})\partial_{1,\alpha}\Psi_i(\vec{s},\vec{r}_{2:N})^*, 
  \end{align}
  using the definition of the weak derivative and Fubini's theorem.
  By Theorem~\ref{thm:factorization}, 
  \begin{equation}
    \tilde{T}_\alpha(\vec{r},\vec{r}) = \frac{1}{2} \sum_i p_i \int
    |\partial_{1,\alpha}\Psi_i(\vec{r},\vec{r}_{2:N})|^2 d\vec{r}_{2:N}
  \end{equation}
  for almost all $\vec{r}$. For any compact
  $K\subset\RR^3$, integration yields
  \begin{equation}
    \int_K d\vec{r} \tilde{T}_\alpha(\vec{r},\vec{r}) = \frac{1}{2}  \sum_i p_i
    \|\partial_{1,\alpha}\Psi_i\|^2_{L^2(K\times\RR^{3N-3})}. 
  \end{equation}
  Since $T_\alpha$ is positive semidefinite, the left hand side is the
  trace of $\frac{1}{2} \partial_{1,\alpha}\partial_{2,\alpha}D$.  Thus, $D$ has
  locally finite kinetic energy if and only if any representing
  $\Gamma\mapsto D$ has locally finite kinetic energy.

\end{proof}

\subsection{Proof of Theorem\;\ref{thm:currents}}

%We restate Theorem~\ref{thm:currents} and prove it:
%\begin{theorem}\label{thm:currents-app}\note{Copy formulations from
%    Sec. 2 for final version. Maybe proof should be more detailed.}
%  Let $D\in \mathcal{D}_{N,1}$, and suppose that for a $G\in
%  L^2(\RR^3\times\RR^3)$,
%  \begin{equation}
%    D(\vec{r},\vec{s}) = (G^\dag * G)(\vec{r},\vec{s}).
%  \end{equation}
%  Suppose also that $\nabla_2 G \in L^2(\RR^3\times \RR^3_\loc)$, so that $D$ has locally
%  finite kinetic enerfy. Then for every $\Gamma\in
%  \mathcal{D}_N$ with $D=D_\Gamma$ a.e. $[d\vec{r}d\vec{s}]$, $\vec{c}_\Gamma \in
%  L^2_\loc(\RR^3)$ and
%  \begin{equation}
%    \vec{c}_\Gamma(\vec{r}) = ([-\rmi\nabla_2 G]^\dag * G)(\vec{r},\vec{r}) \quad\text{a.e.}[d\vec{r}].
%  \end{equation}
%\end{theorem}
\begin{proof}
  Most of the proof is similar that of 
  Theorem~\ref{thm:lockin}, so we skip some details.

  Let $\Gamma$ be such that $D_\Gamma = D$ almost everywhere. The state $\Gamma$
  has a locally finite kinetic energy by
  Theorem~\ref{thm:lockin}. By a reasoning similar to that of
  the proof of this lemma, we obtain 
  \begin{equation}
    c_{\Gamma,\alpha}(\vec{r}) = c_\alpha(\vec{r}) = [\diag
    (-\rmi\partial_{2,\alpha} G)^\dag\ast G](\vec{r},\vec{r}) 
  \end{equation}
  almost everywhere, independently of $\Gamma$. Taking the absolute
  value, integrating over a compact $K\subset\RR^3$ and applying the
  Cauchy--Schwarz inequality,  we obtain the bound 
  \begin{align}
    \int_K |c_\alpha(\vec{r})| \mathrm d\vec{r} &\leq \|\partial_{2,\alpha}
    G\|_{L^2(\RR^3\times K)} \|G\|_{L^2(\RR^3\times K)} \nonumber \\ &< +\infty.
  \end{align}
\end{proof}

%\bibliography{refs.bib}
%\bibliographystyle{unsrt}  % options: plain, unsrt, alpha, abbrv,
			     % apalike, siam, acm

%

\end{document}